\documentclass[11pt]{article}
\usepackage[utf8]{inputenc}
\usepackage[top = 1in, bottom = 1in, left = 1in, right =1in]{geometry}
\usepackage{latexsym,amsmath,amssymb,amsthm,graphicx}
\usepackage{algorithm,algorithmicx,algpseudocode}
\usepackage{enumitem}
\usepackage{url}
\usepackage{mathtools}
\usepackage{xcolor}
\usepackage{subcaption}
\usepackage{bm}
\usepackage{soul}
\usepackage{bbold}
\usepackage{dsfont}
\usepackage{appendix}
\usepackage{cite}
\usepackage[mathscr]{euscript}
\usepackage[makeroom]{cancel}
\usepackage{multirow}
\usepackage[bottom]{footmisc}
\allowdisplaybreaks
\makeatletter
 
\newcommand{\pushright}[1]{\ifmeasuring@#1\else\omit\hfill$\displaystyle#1$\fi\ignorespaces}
\newcommand{\pushleft}[1]{\ifmeasuring@#1\else\omit$\displaystyle#1$\hfill\fi\ignorespaces}
\makeatother

\providecommand{\keywords}[1]{\textbf{{Keywords --}} #1}

\setlist[itemize]{leftmargin=*}

\newcounter{thm_count}

\theoremstyle{remark} 
\newtheorem*{theorem*}{Theorem} 
\newtheorem{theorem}[thm_count]{\bf Theorem} 

\newcounter{assump}
\newcounter{defn}

\newcounter{expl}
\newcounter{rmrk}
\theoremstyle{remark}

\newtheorem{assumption}[assump]{\bf Assumption}
\newtheorem{definition}[defn]{\bf Definition}
\newtheorem{remark}[rmrk]{\bf Remark}
\newtheorem{example}[expl]{\bf Example}
\renewenvironment{proof}[1][Proof:]{\begin{trivlist}
\item[\hskip \labelsep {\bfseries #1}]}{\end{trivlist}}

\title{{Private Learning on Networks} \footnote{This research is supported in part by National Science Foundation awards 1421918 and 1610543, and Toyota InfoTechnology Center. Any opinions, findings, and conclusions or recommendations expressed here are those of the authors and do not necessarily reflect the views of the funding agencies or the U.S. government.}}

\author{Shripad Gade $\qquad$ Nitin H. Vaidya \\ \\ Department of Electrical and Computer Engineering, and \\ Coordinated Science Laboratory, \\ University of Illinois at Urbana-Champaign. \\ Email: \{gade3, nhv\}@illinois.edu \\ \\ Technical Report}

\date{}

\begin{document}

\maketitle

\begin{abstract}
Continual data collection and widespread deployment of machine learning algorithms, particularly the distributed variants, have raised new privacy challenges. In a distributed machine learning scenario, the dataset is stored among several machines and they solve a distributed optimization problem to collectively learn the underlying model. We present a secure multiparty computation inspired privacy preserving distributed algorithm for optimizing a convex function consisting of several possibly non-convex functions. Each individual objective function is privately stored with an agent while the agents communicate model parameters with neighbor machines connected in a network. We show that our algorithm can correctly optimize the overall objective function and learn the underlying model accurately. We further prove that under a vertex connectivity condition on the topology, our algorithm preserves privacy of individual objective functions. We establish limits on the what a coalition of adversaries can learn by observing the messages and states shared over a network.
\end{abstract}


\section{Introduction}

Advances in fast machine learning algorithms have resulted in widespread deployment of machine learning algorithms \cite{abadi2016deep}. Distributed learning and inference have become popular due to their inherent efficiency, scalability, robustness and geo-distributed nature of datasets \cite{boyd2011distributed,kraska2013mlbase,NIPS2014_5597,Li2013b,cano2016towards,Nedi2015a,gade16distoptclientserver}. Distributed learning reduces communication requirements of learning, since machines communicate/share updates (gradients or states) that are much smaller in size than datasets. Several distributed optimization algorithms have appeared in literature over the past decade \cite{nesterov2012efficiency,liu2015asynchronous,agarwal2011distributed,Singh2014,Nedic2007,nedic2011asynchronous,Nedi2015,zhang2014asynchronous,huang2015differentially,rabbat2004distributed,zhu2012distributed,gade16distoptclientserver,sra2012optimization}. Solutions to distributed optimization of convex functions have been proposed for myriad scenarios involving directed graphs \cite{Nedi2015}, link failures and losses \cite{hadjicostis2016robust}, asynchronous communication models \cite{nedic2011asynchronous,wei20131,zhang2014asynchronous}, stochastic objective functions  \cite{agarwal2011distributed,ram2010distributed,ram2009incremental}, fault tolerance \cite{su2016fault} and differential privacy \cite{huang2015differentially}.

Privacy has emerged to be one of the most critical challenges in machine learning \cite{shokri2015,abadi2016deep,abbe2012privacy,alizadeh2012grid,pasqualetti2012cyber}. For instance, in the healthcare industry, hospitals/insurance providers use medical records to learn predictive models estimating individual risk towards certain diseases. Data driven learning on DNA sequences and medical records, has allowed to predict individual predisposition to certain life threatening ailments, based solely on genetic factors. The genomic data and medical records are highly sensitive and need to be protected (they can be easily misused by malicious entities) \cite{shokri2015privacy}. Personalized recommendation systems is another application in which privacy is important. Specifically, predicting the next word a user will input (using language models trained on personal typing history e.g. SwiftKey \cite{SwiftKey}) or recommending pictures/videos based on viewing history are examples where sensitive personal preferences are used for training models, and may threaten privacy \cite{mcmahan2016federated}. Corporations use data based analytics to learn models for everything ranging from product demand-supply schedules \cite{hippert2001neural}, pricing, to marketing/advertising strategies using consumer data. Collaborative shipping and transportation of goods among sellers can be modeled as a linear programming problem, where, participating sellers would want to protect privacy \cite{vaidya2009privacy,hong2016privacy}. Financial fraud detection, transaction authentication, stock prediction, and credit risk etc can be learned from large financial databases available with banks and credit bureaus. Roboticists are interested in looking at problems like coverage \cite{cortes2005spatially}, rendezvous \cite{ren2005survey}, search and tracking \cite{olfati2007distributed,gade2013heterogeneous}, herding \cite{gade2015herding} etc. that involve use of multiple, collaborative robots. Robots location, internal states and observations may be extremely sensitive (based on the application) and one would ideally want robots to cooperatively solve problems, without sharing any of these pieces of private information. All of these critical machine learning applications utilize datasets that contain personal and often times extremely sensitive information. In this report we address a fundamental question - \textit{``How do we accurately learn underlying models without leaking any private information?"}

Current privacy preserving methods can be broadly classified into cryptographic approaches and non-cryptographic approaches \cite{weeraddana2013per}. Cryptographic approaches as the name suggests, use cryptographic techniques and have been extensively studied in literature \cite{weeraddana2013per,weeraddana2012privacy,goldreich2009foundations,goldwasser1997multi,pinkas2002cryptographic,hong2016privacy,bednarz2012methods}. Cryptography based privacy preserving techniques typically provide better security, however, are computationally expensive and inefficient \cite{bednarz2012methods}. Cryptographic methods are also vulnerable to attacks that involve stealing of encryption keys \cite{weeraddana2013per}. 

Several types of non-cryptographic approaches have gained popularity in recent years. \emph{$\epsilon$-differential privacy} is a popular probabilistic technique that involves use of randomized perturbations. Differential privacy aims to maximize accuracy of the queries made to a statistical database while minimizing the probability of information leakage \cite{huang2015differentially,7431982,nozari2015differentially,abadi2016deep}. Differential privacy methods suffer from a fundamental trade-off between the accuracy of the solution and the privacy margin (parameter $\epsilon$). \emph{Transformation techniques} involve concocting a new problem via algebraic transformations such that the solution of the new problem is the same as the solution of the old problem. This enables agents to conceal private problem data (of the original problem) effectively while the quality of solution is preserved \cite{weeraddana2013per,weeraddana2012privacy}. Scaling, translation, affine transformation and certain nonlinear transformations are some of the algebraic transformation techniques that have been used in literature \cite{mangasarian2012privacy,wang2011secure,dreier2011practical,weeraddana2012privacy,weeraddana2013per}. 

Lou \emph{et.al} in \cite{lou2015privacy} claim that unconstrained, synchronous, distributed optimization protocols are generally not privacy preserving and that asynchronous updates and projection sets (constrained optimization) can be used to introduced privacy. We, however, show in Example~\ref{Ex:SimAtt}, that projection sets do not provide privacy protection from adversaries\footnote{We show that, in a system implementing projected gradient descent protocol from \cite{ram2010distributed} or \cite{gade16convsum}, a strong adversary is successful in uncovering private objective functions.}. We also conjecture, that if adversary can additionally observe and keep track of state updates then asynchronous nature of algorithm may not be effective in providing any additional privacy.  

Problem structure or the learning system architecture can be further exploited to improve privacy. Preliminary ideas on such strategies, viz. function partitioning (splitting) and function sharing are reported in our previous report \cite{gade16convsum}. We also show that structured randomization of gradient updates can improve privacy in client-server architecture (multiple parameter servers and multiple clients) in our prior report \cite{gade16distoptclientserver}. In this report, we present the function sharing approach to privacy preservation followed by convergence (correctness) and privacy results. The objective of this work is to introduce and quantify privacy guarantees in distributed machine learning.

\subsection{Contributions}
The main contributions of the report are threefold. 
\begin{enumerate}
\item We develop a privacy preserving distributed optimization algorithm using a function sharing strategy. Our novel algorithm is easy to use and computationally inexpensive.
\item We present and rigorously prove that the algorithm accurately optimizes the objective function while ensuring privacy of individual objective functions. 
\item The strategy presented here (and ideas in \cite{gade16convsum}), allows us to establish crafty privacy preserving techniques, that exploit the structure of distributed learning architectures. We believe that such methods are practical for introducing privacy in distributed learning. 
\end{enumerate}

\subsection{Organization}
Problem formulation, adversary model and privacy definitions are presented in Section~\ref{Sec:ProbForm}. Our privacy preserving distributed optimization protocol in presented and detailed in Section~\ref{Sec:Algo}. Convergence and privacy guarantees are presented and proved in Section~\ref{Sec:ConvPrivProof}. Simulation results are presented in Section~\ref{Sec:Results}.

\subsection{Notation}
Let the set of agents (also referred to as nodes) be denoted by $\mathcal{V}$ and the number of agents be $|\mathcal{V}| = S$. We use the symbol ``$\sim$" to denote a directed communication link over which information sharing can occur between agents, e.g., $I \sim G$ denotes a directed communication link from agent $I$ to agent $G$. We define edge set as the set of all directed communication links $\mathcal{E} = \{(u,v) : u \in \mathcal{V}, v \in \mathcal{V} \text{ and } u \sim v \}$. We consider bidirected communication links, hence, edge $(v,u) \in \mathcal{E}$ whenever edge $(u,v) \in \mathcal{E}$. Communicating links between agents induce a graph $\mathcal{G} = (\mathcal{V}, \mathcal{E})$, defined by the set of all nodes (agents), $\mathcal{V}$ and the set of all edges (communication links), $\mathcal{E}$. The neighborhood set of agent $J$, the set of all agents that communicate with agent $J$, is denoted by $\mathcal{N}_J$. $\|.\|$ denotes Euclidean norm for vectors.

Every agent maintains an estimate of the model parameter vector (also referred to as iterate). Iterate stored in agent $I$ at iteration $k$ is denoted by $x^I_{k}$, where the superscript denotes the agent-id, the subscript denotes the time index. The average of iterates at time instant $k$ is denoted by $\bar{x}_{k}$ and the disagreement of an iterate ($x^J_{k}$) with the iterate average ($\bar{x}_{k}$) by $\delta^J_{k}$. 
\begin{equation}
\bar{x}_{k} = \frac{1}{S} \sum_{J=1}^S x^J_{k}, \qquad \delta^J_{k} = x^J_{k} - \bar{x}_{k}. \label{Eq:deltaDef}
\end{equation}
Agents also maintain an estimate of iterate average denoted by $v^J_k$, where the superscript represents the agent-id and the subscript denotes the iteration index.

\section{Problem Formulation} \label{Sec:ProbForm}
We consider a distributed optimization problem involving $S$ agents, each of whom has access to a private, possibly non-convex objective function $f_i(x)$. Agents intend to collectively solve the following optimization problem,
\begin{align}
\text{Find} \quad x^* \in \underset{x \in \mathcal{X}}{\text{argmin}} \; f(x), \label{Eq:OptProb}
\end{align}
where $\mathcal{X}$ is the feasible parameter set, and $f(x) \triangleq \sum_{i=1}^S f_i(x)$ is a convex objective function. The dimension of the problem (number of parameters in the decision vector, $x$) is denoted by $D$. We enforce the following assumption on the functions $f_i(x)$ and on the feasible parameter set $\mathcal{X}$.
\begin{assumption} [Objective Function and Decision Set] \label{Asmp:FunSet}
The individual objective functions $f_i(x)$ and the feasible parameter set $\mathcal{X}$ satisfy the following properties,
\begin{enumerate}[label=(\Alph*)]
\item  The objective functions $f_i : \mathbb{R}^D \rightarrow \mathbb{R}, \; \forall \; i = 1, 2,{ }\ldots, S$ are potentially non-convex functions of model parameter vector $x$. However, the sum of individual objective functions is necessarily convex, i.e., $f(x) := \sum_{i=1}^S f_i(x)$ is a convex function. \label{Asmp:Function} 
\item The feasible parameter vector set, $\mathcal{X}$, is a non-empty, convex, and compact subset of $\mathbb{R}^D$. \label{Asmp:Set}
\end{enumerate}
\end{assumption}
\noindent Following the above assumption, we will refer to the aggregate function, $f(x)$, as a convex aggregate of non-convex functions. 

We further make a boundedness assumption on the gradient of individual function $f_i(x)$ followed by an assumption on the Lipschitzness of gradients. 
\begin{assumption} [Gradient Boundedness and Lipschitzness] \label{Asmp:GradientCond} Let $g_h(x)$ denote the gradient of the function $f_h(x)$. The gradients $g_h(x)$ satisfy,
\begin{enumerate} [label=(\Alph*)]
\item The gradients are norm bounded, i.e. there exist scalars $L_1, L_2, \ldots, L_S $ such that, $\| g_h(x) \| \leq L_h; \; \forall \ h \;( = \{1, 2, \cdots, S\}) \; \text{and} \; \forall \ x \in \mathcal{X}$.\label{Asmp:SubBound}
\item Each function gradient ($g_h(x)$) is assumed to be Lipschitz continuous i.e. there exist scalars $N_h > 0$ such that, $\|g_h(x) - g_h(y) \| \leq N_h \| x -y\|$ for all $x \neq y$ ($x,y \in \mathcal{X}$) and $\forall \ h \; (= \{1, 2, \ldots, S\})$. \label{Asmp:GradLip}
\end{enumerate}
\end{assumption}

Agents communicate with their neighbors and share model parameter estimates. The communication graph $\mathcal{G}$ constitutes bidirectional links. $\mathcal{G}$ is assumed to be a connected graph. All agents are assumed to be synchronous and fault-free. All communication links are assumed to be reliable. Throughout this report, we will use the following definitions and notation regarding the set of all optima ($\mathcal{X}^*$) and the function value at optima ($f^*$),
$$ f^* = \inf_{x \in \mathcal{X}} f(x), \qquad \mathcal{X}^* = \{x \in \mathcal{X} | f(x) = f^*\}, \qquad dist(x, \mathcal{X}^*) = \inf_{x^* \in \mathcal{X}^*} \|x - x^*\|.$$ 
\noindent The optimal function value, at the solution of the optimization problem or the minimizing state vector is denoted by $x^*$, is denoted by $f^*$.

\subsection{Adversary Model} \label{Sec:AdvModel}
Various adversary models have been studied and employed in literature and they are broadly classified, as follows, based on the capabilities and intentions of the adversary:
\begin{itemize}
    \item Passive Adversary \cite{goldwasser1997multi} - Passive adversaries limit the interaction to eavesdropping on their own communication channels and storing evolution of observed states and other information. Passive adversaries follow default protocol as other ``good" participants\footnote{An agent that is not an adversary is referred to as ``good" agent.}. 
    \item Active Adversary (Byzantine Adversary \cite{goldwasser1997multi,lamport1982byzantine}) - Active adversary model empowers the adversary to send, tamper and delete parts of messages being exchanged. They may arbitrarily deviate from the default protocol. 
    \item Curious Adversary - Curious adversaries try to uncover private information of the system (e.g. individual objective functions $f_i(x)$) using information available with them. 
    \item Bounded Rationality - Bounded rationality adversaries have limited computational power and they cannot perform complex computations. Conversely, unbounded rationality models have also been proposed and there exist strategies like Shamir's secret sharing mechanism \cite{shamir1979share} that are secure against adversaries with unbounded rationality (information theoretic security). 
\end{itemize}

In this report, we model the adversary as \emph{Passive-Curious} (PC) entity\footnote{For the purposes of this report, we will assume that all adversaries are strong PC adversaries.}. An adversary (denoted by $\mathscr{A}$) is PC, if it records the evolution of the system states by eavesdropping and tries to uncover information that is private to other agents. We consider a very strong adversary in the sense that the adversary has access to a lot of information about the system, which is typically unavailable in a distributed setting. A strong adversary makes privacy preservation difficult and underlines the strength of our approach. Note that we consider privacy in synchronous setting and it further shows the competence of our approach\footnote{Privacy in synchronous executions is relatively challenging as compared to asynchronous executions \cite{li2015differentially}. This follows from the fact that in asynchronous executions the step sizes (and the update counts) are uncoordinated and it makes estimating gradients form the observed states difficult.}. The assumption of a strong and motivated adversary with complete knowledge of the system (states and the underlying network) is essential to ensure that we perform vulnerability analysis in the worst scenario, following \textit{Kerckhoffs's principle}.  

Adversary has access to all the states of the system (model parameters $x^I_k$ of all agents $I$ at all time instants $k$) and its own private objective function. It also has access to all incoming and outgoing secure communications from itself. The adversary also has an understanding of the network structure and topology. The adversary, being passive, is restricted to following the same distributed optimization protocol as other ``good" agents\footnote{Agent $I (\in \mathcal{V})$ is ``good" agent if $I \cancel{\in} \mathcal{A}$.}. The adversary, although curious, wants the system to correctly solve the optimization problem. 

An alliance of cooperating adversaries is called a coalition (denoted by $\mathcal{A}$). In this report we also show privacy guarantees against a coalition of $f$ PC adversaries ($|\mathcal{A}| = f$). Information available to any adversary $\mathscr{A} \in \mathcal{A}$, is shared instantaneously among the other coalition members. Hence, any adversary $\mathscr{A}$ can observe the evolution of all states, the network topology and any communication that is inbound and outbound from any of the adversaries. 

\subsection{Privacy Definitions} 
We define privacy as the inability of the adversary to uncover private objective functions. Distributed optimization protocols involve agents utilizing its own local gradients to perform state updates followed by sharing the state estimate with neighbors. Since, a PC adversary has access to both the states and the network topology, under distributed gradient descent protocols \cite{ram2010distributed}, a PC adversary can estimate states and local gradients at these states. An adversary, by employing numerical techniques like polynomial regression (polynomial interpolation, numerical integration etc.) on the history of state and gradients pair ($x, \nabla f_i(x)$), can estimate gradients of private objective functions private to agents. This does not exactly provide $f_i(x)$, since $\nabla f_i(x)$ is available only at discrete state values (and not known analytically) and there will be integration constants that cannot be resolved. However, an adversary can guess with good accuracy the form (shape) and general behavior of a private function. 

In this report we attempt to solve a more demanding privacy problem. If we assume that an adversary can in fact estimate the gradients accurately, it is easy to see that the adversary can uncover private objective functions. We provide a scheme that ensures that no adversary can guess the individual objective functions with any accuracy. Formally we first define the set of all admissible private objective functions $f_i(x)$ as,
\begin{definition}[Admissible Function Set] \label{Def:AdmFunSet}
The set of all possible objective functions $f_i(x)$ is called the admissible function set and is denoted by $\mathcal{F}$. 
\end{definition}
\noindent Let $\mathcal{F}^S$ denote the set of $S$ private objective functions (each associated to one of the $S$ agents and each belonging to set $\mathcal{F}$). Hence, $\mathcal{F}^S = \{(f_1(x), f_2(x), \ldots, f_S(x))\}$, where $f_i(x) \in \mathcal{F}$ for $i = \{1, 2, \ldots, S\}$. 

Following the definition of admissible function set, we use the inability of the adversary to accurately guess the function from set $\mathcal{F}^S$ as the basis for defining privacy. Formally, an optimization protocol is called privacy preserving if it satisfies the following definition of privacy.
\begin{definition}[Privacy] \label{Def:Privacy}
An optimization algorithm is said to be privacy preserving under a given adversary model, if the information observed by the adversaries is compatible with any set of functions i.e. any element of $\mathcal{F}^S$ (e.g. ($h_1(x), h_2(x), \ldots, h_S(x)$) $\in \mathcal{F}^S$), such that the sum of its functions is original aggregate function $f(x)$, i.e. $f(x) = \sum_{i=1}^S h_i(x)$. 
\end{definition}
\noindent Intuitively, we define privacy as the inability of a PC adversary (or a coalition) to reduce the ambiguity associated with any guessed objective function. In other words, even after observing the execution of the optimization protocol (states, gradients etc.), an adversary (or a coalition) finds any arbitrarily guessed candidate functions (belonging to $\mathcal{F}^S$ and such that its functions add up to original $f(x)$) to be equally likely of being the true objective function.
 
\subsubsection*{A Motivating Example}

In the following example we show that, in a standard distributed optimization problem being solved by a distributed protocol, a strong PC adversary can guess private objective functions merely by observing the state evolution.  

\begin{example} \label{Ex:SimAtt}
Let us consider a distributed optimization problem over a network with three nodes in a fully connected topology (Figure~\ref{Fig:F-0}). The private objective functions for each of the nodes are $f_1(x) = (x-1)^2, \ f_2(x) = (x-2)^2 + (x-2)^4 \text{ and } f_3(x) = (x-3)^4$. The system of agents tries to find an optima of $f(x) = f_1(x)+f_2(x)+f_3(x)$ using distributed optimization protocol (for deterministic objective functions) from \cite{gade16convsum} or \cite{ram2010distributed}. Agent 1 is a strong PC adversary (as shown in Figure~\ref{Fig:F-0} with two concentric circles in blue) with private objective function $f_1(x)$. The adversary has complete knowledge of state evolution of all states ($x^1_k, x^2_k, x^3_k$) and the network connectivity. Agent 1 can estimate both $(x^2_k, \nabla f_2(x^2_k))$ and $(x^3_k, \nabla f_3(x^3_k))$, since it is privy to the states, network structure and step sizes. The adversary can estimate the gradient function using polynomial interpolation or polynomial regression. Integrating the gradient provides the objective functions $f_2(x)$ and $f_3(x)$ with an ambiguity of a constant term. 

In this example we show a successful attack from a strong PC adversary. Distributed optimization protocol from \cite{gade16convsum} is executed for 300 iterations. The system of agents correctly solves the optimization problem. During the execution, adversary observes the state evolution and using its knowledge of the protocol and the network structure, it estimates gradient values at discrete iterates. Agent 1 uses least squares polynomial fitting (\texttt{numpy.polyfit}) to guess the gradient function based on estimated gradient and iterate pairs. This gradient function is further integrated to obtain the objective function. In this attack, the functions guessed by the adversary are tabulated in Table~\ref{Tab:AttackEx}. The adversary is successful in uncovering objective functions $f_2(x)$ and $f_3(x)$ (with an ambiguity in the constant term).
\begin{table}[h]
\centering
\begin{tabular}{ |l|l|l|l| }
\hline
\multicolumn{4}{ |c| }{\textbf{Problem 1}} \\
\hline
$\bm{f_i(x)}$ & \textbf{True Obj. Function}, $\bm{f_i(x)}$ & \textbf{Estimated Obj. Function}, $\bm{\tilde{f}_i(x)}$ & \textbf{Ambiguity} \\ \hline
$f_1(x)$ & $x^2 - 2x + 1$ & Known to Adversary & - \\ \hline
$f_2(x)$ & $x^4 - 8x^3 + 25x^2 -36 x + 20$ & $x^4 - 8x^3 +25x^2 -36x + C_2$ & $C_2$ \\ \hline
$f_3(x)$ & $x^4 - 12x^3 + 54x^2 - 108 x + 81$ & $x^4 - 12x^3 + 54x^2 -108x + C_3$ & $C_3$ \\ \hline
\end{tabular}
\caption{An example showing successful attack by a PC adversary.}
\label{Tab:AttackEx}
\end{table}
\end{example}

\begin{figure}[!b]
\begin{subfigure}{.48\textwidth}
  \centering
  \includegraphics[width=0.93\linewidth]{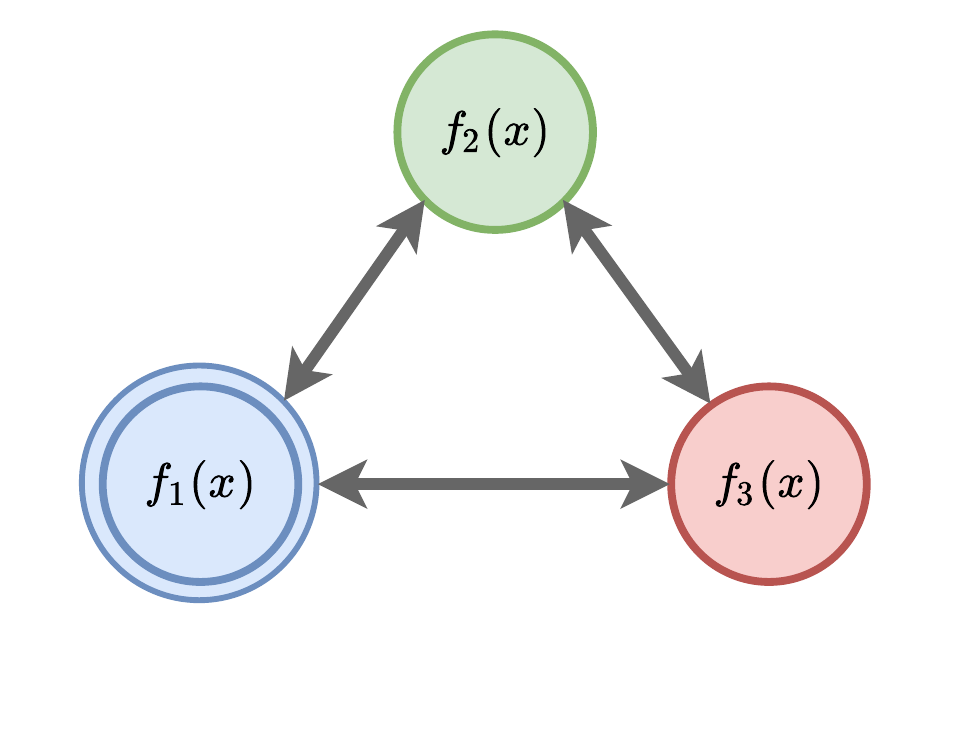}
  \caption{Adversary (Agent 1) is shown with two concentric circles (blue).}
  \label{Fig:F-0}
\end{subfigure} \hfill
\begin{subfigure}{.48\textwidth}
  \centering
  \includegraphics[width=0.95\linewidth]{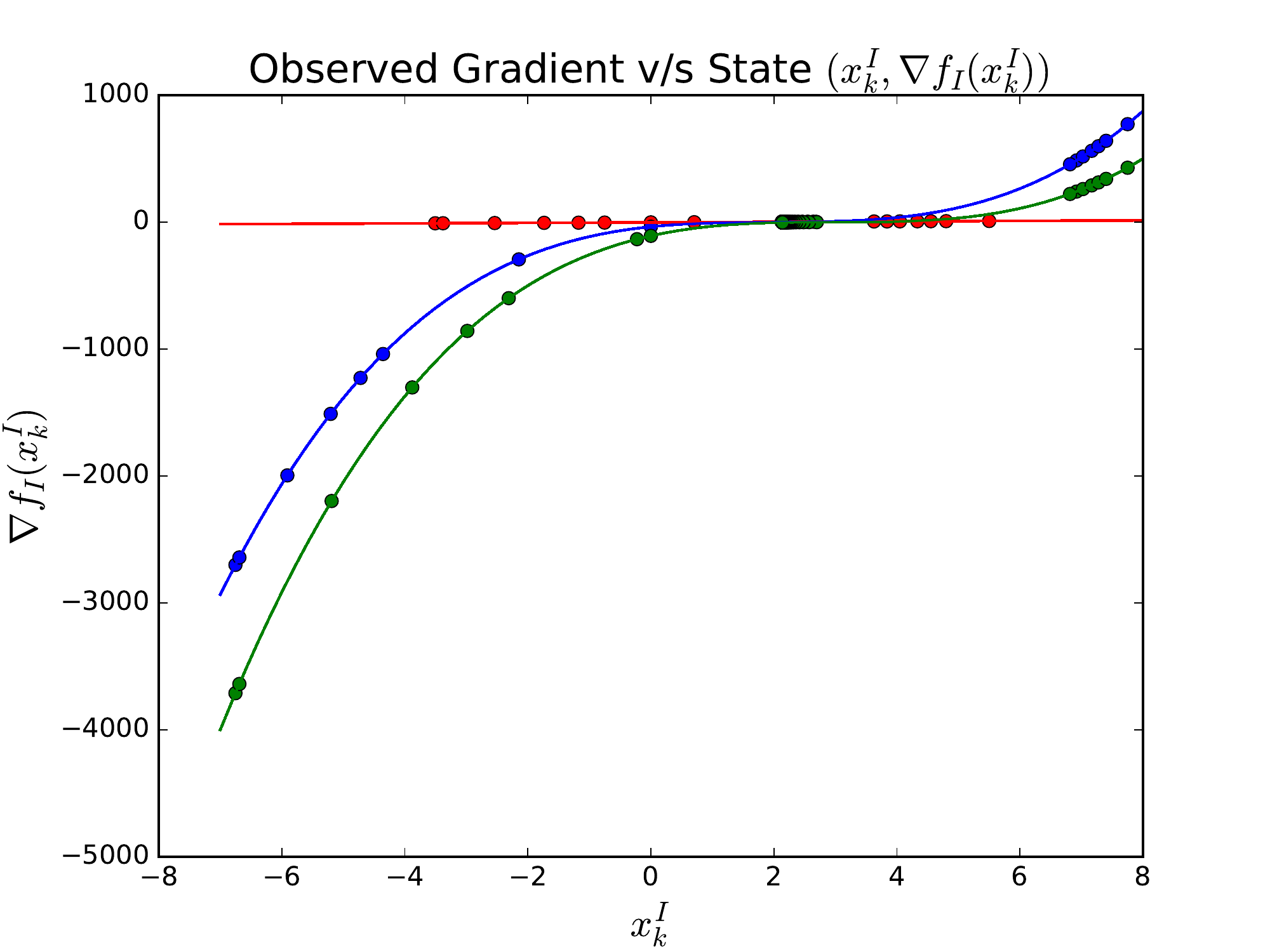}
  \caption{Observed gradients plotted against iterates. Figure also shows least squares polynomial fitting.}
  \label{Fig:F-01}
\end{subfigure} \\
\begin{subfigure}{.48\textwidth}
  \centering
  \includegraphics[width=.95\linewidth]{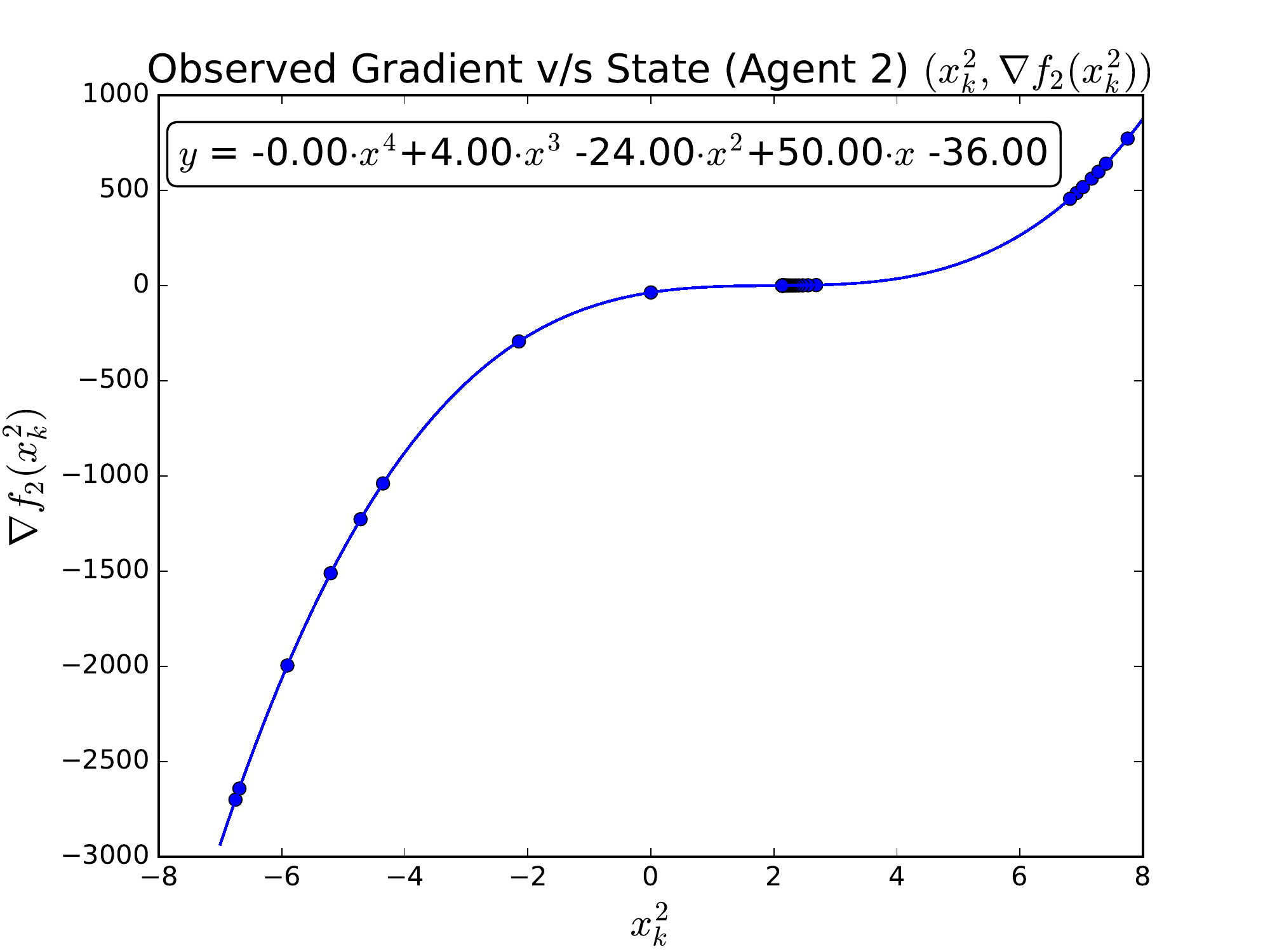}
  \caption{Least squares polynomial fit for $\nabla f_2(x)$.}
  \label{Fig:F-1}
\end{subfigure} \hfill
\begin{subfigure}{.48\textwidth}
  \centering
  \includegraphics[width=.95\linewidth]{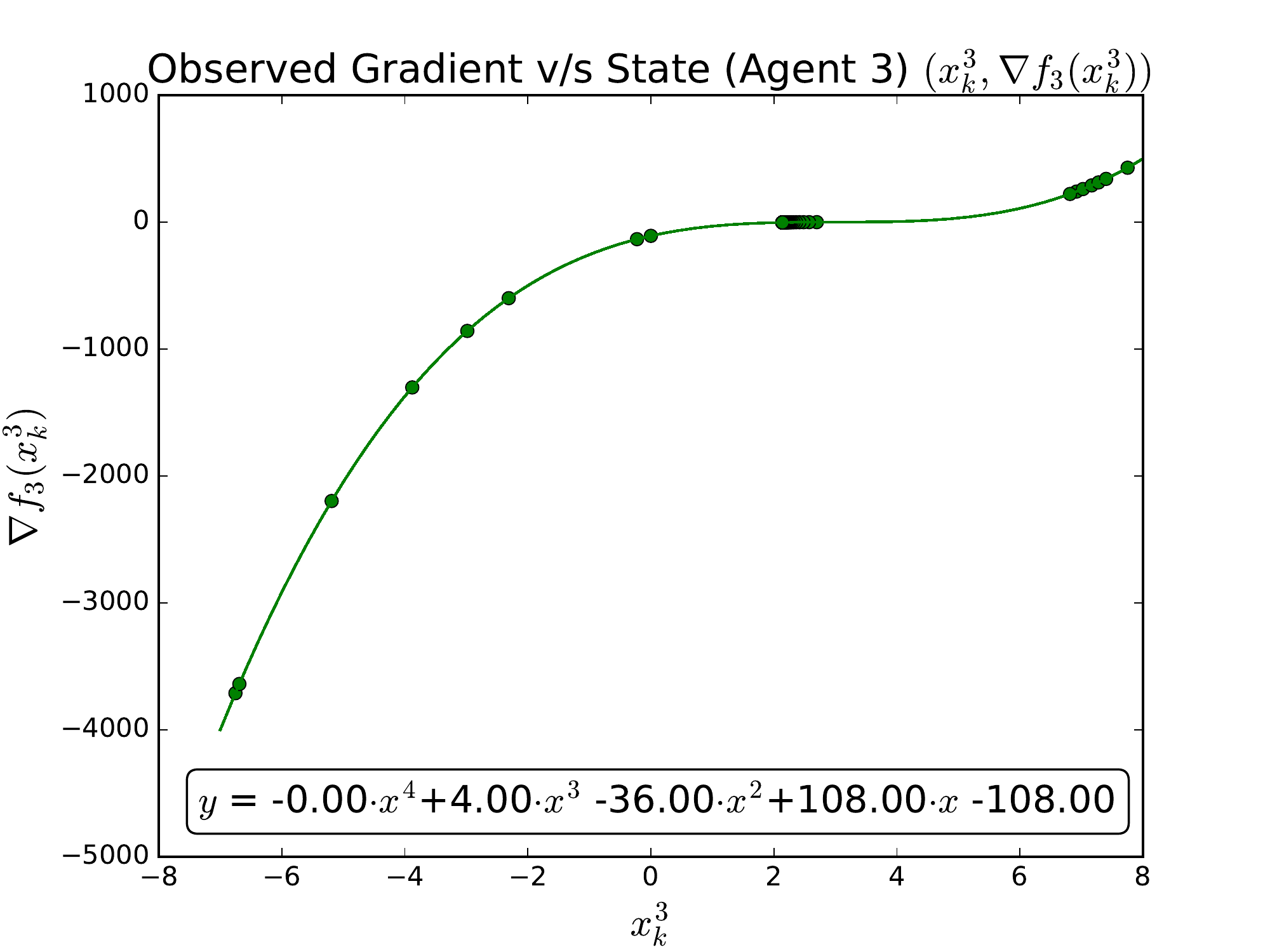}
  \caption{Least squares polynomial fit for $\nabla f_3(x)$.}
  \label{Fig:F-2}
\end{subfigure}%
\caption{An example demonstrating a successful attack by a strong PC adversary and the subsequent privacy loss in distributed optimization.}
\label{Fig:ExamplePLoss}
\end{figure}

\section{Privacy Preserving Distributed Optimization Algorithm} \label{Sec:Algo}

We present privacy preserving distributed optimization algorithm inspired by secure multiparty aggregation algorithm in \cite{abbe2012privacy}. We present a two step protocol. In the first step we perform a secure exchange of arbitrarily generated functions and then use these functions to transform the private objective functions. The second step involves an interleaved consensus and projected gradient algorithm \cite{gade16convsum}. Our algorithm, differs from popular distributed gradient descent protocols, in the sense that it involves the objective function being transformed enabling us to provide privacy guarantees. The convergence analysis used here (presented in \cite{gade16convsum}) is also novel in the sense that it proves convergence of a distributed projected gradient descent algorithm for a convex aggregate of non-convex functions.

Each agent is endowed with a private objective function $f_i(x)$. Now consider that each agent, $I$, arbitrarily generates functions $R_{I,J}(x)$ corresponding to each of its neighbors ($J \in \mathcal{N}_I$). All agents, $I$, then securely share $R_{I,J}(x)$ with their neighbors $J \in \mathcal{N}_I$. Any agent $I$, thus has access to functions that it has shared with neighbors, $R_{I,J}(x)$ for $J \in \mathcal{N}_I$, and functions that it has received from neighbors, $R_{K,I}(x)$ for $K$ such that $I \in \mathcal{N}_K$. This step is followed by obfuscation of the private objective function using the arbitrarily generated functions $R_{I,J}(x)$. Every agent transforms its objective function by adding all the arbitrary functions that it has received from its neighbors and subtracting all the arbitrary functions transmitted to the neighbors, and arrives to the obfuscated objective function $\hat{f}_i(x)$, 
\begin{align}
\hat{f}_i(x) = f_i(x) + \sum_{K:I \in \mathcal{N}_K} R_{K,I}(x) - \sum_{J \in \mathcal{N}_I} R_{I,J}(x)
\label{Eq:F_New}
\end{align}

The objective obfuscation step is followed by iterative distributed projected gradient algorithm similar to \cite{ram2010distributed,gade16convsum}. The algorithm is formally presented as Algorithm~\ref{Algo:PrivDistOptNCFun}. The interleaved consensus and projected gradient algorithm involves two operations. The first operation in the algorithm is to fuse information from the neighbors and build an estimate of average of the parameter vector. A doubly stochastic matrices $B_k$, with the property that any entry $B_k[I,J]$ is greater than zero if and only if $I$ and $J$ can communicate with each other, is used for information fusion. Also, we assume that all non-zero entries are lower bounded by $\eta$, i.e. if $B_k[I,J] > 0$ then $B_k[I,J] \geq \eta$ for some constant $\eta>0$. We can write the fusion step as,
\begin{align}
    v^J_{k} = \sum_{I \in \mathcal{N}_J} B_k[J,I] x^I_{k}. 
    \label{Eq:InfFusALG1}
\end{align}

The information aggregation step is followed by projected gradient step and it is formally written as, 
\begin{align}
    x^J_{k+1} = \mathcal{P}_{\mathcal{X}}\left[v^J_{k} - \alpha_k \ \nabla f_J(v^J_{k}) \right]. \label{Eq:ProjGradALG2}
\end{align}
where, $\mathcal{P}_{\mathcal{X}}$ is the projection operator onto set $\mathcal{X}$ (cf. \cite{bertsekas2003convex}), and $\alpha_k$ is the learning rate.
Projected gradient descent is a well known iterative gradient based method that guarantees convergence to optimum under reducing learning rate ($\alpha_k$) \cite{bertsekas1976goldstein}. We assume that the monotonically non-increasing learning rate (step-size, $\alpha_k$) possesses the following properties,
\begin{align}
     \alpha_k > 0, \ \forall k \geq 0; \quad \alpha_{k+1} \leq \alpha_k, \ \forall k \geq 0; \quad  \sum_{k=0}^\infty \alpha_k = \infty; \ \text{and} \quad \sum_{k=0}^\infty \alpha_k^2 < \infty. \label{Eq:LearnStepCond}
\end{align}

\begin{algorithm}[t]
\caption{Privacy Preserving Distributed Optimization}
\begin{algorithmic}[1]
\State Input: $x^J_k$, $\alpha_k$, NSteps \Comment{NSteps - Termination Criteria}     
\State Result: $x^* = \underset{x \in \mathcal{X}}{\text{argmin}} \sum_{i=1}^{S} f_i(x) $ 
\State Every agent $I$ shares function $R_{I,J}(x)$ with all neighboring agents $J$ 
\State Every agent $I$ updates its objective function to $\hat{f}_i(x)$ \Comment{Eq.~\ref{Eq:F_New}}
\For {k = 1 to NSteps} 
    \For {J = 1 to S}
        \State $v^J_{k} = \sum_{I \in \mathcal{N}_J} B_k[J,I] x^I_{k}$   \Comment{Information Fusion}
        \State $x^J_{k+1} = \mathcal{P}_\mathcal{X} \left[ v^J_{k} - \alpha_k \nabla \hat{f}_J(v^J_{k})\right]$ \Comment{Projected Gradient Step}
    \EndFor
\EndFor 
\end{algorithmic}
\label{Algo:PrivDistOptNCFun}
\end{algorithm}

\begin{remark} [Invariant Aggregate] \label{Rem:InvAgg}
Every arbitrary function $R_{I,J}(x)$ gets added to the objective function of $J$ and gets subtracted from the objective function of $I$. Then, it is easy to see that the sum of obfuscated objective functions, $\hat{f}_i(x)$, is the same as the sum of private objective functions, $f_i(x)$, i.e. the aggregate $f(x)$ is invariant under this transformation.
\begin{align}
\sum_{i=1}^S \hat{f}_i(x) = \sum_{i=1}^S \left( f_i(x) + \sum_{K \ : \ I \in\mathcal{N}_K} R_{K,I}(x) - \sum_{J \in \mathcal{N}_I} R_{I,J}(x) \right) = \sum_{i=1}^S {f}_i(x)  
\end{align} 
\end{remark}

\begin{remark} [Computational Complexity]
Secure transmissions are computationally complex and expensive. It is one of the major drawbacks of cryptographic approaches to privacy and security \cite{weeraddana2013per}. Although, our algorithm involves a secure transmission, it is worth noting that this exchange happens only once, at the start of the algorithm. The iterative consensus and projected gradient steps (Lines 5 - 10 of Algorithm~\ref{Algo:PrivDistOptNCFun}) belong to the category of standard iterative distributed optimization algorithms that are known to be computationally inexpensive. Our algorithm does not involve any additional computational overhead due to privacy.
\end{remark}

\subsection{Generating Arbitrary Functions $R_{I,J}(x)$} \label{Sec:FunDetails}
We claim that using obfuscated functions $\hat{f}_i(x)$ instead of original functions $f_i(x)$, the agents expose only $\hat{f}_i(x)$ and can hide the original objective functions. We will now build intuition for selecting the arbitrarily generated functions $R_{I,J}(x)$. We know that any agent $I$ has access to the following quantities: Original objective function, $f_i(x)$, functions shared by agent $I$ to neighbors, $R_{I,J}(x)$ (controlled by agent $I$), and functions received by agent $I$ from neighbors, $R_{K,I}(x)$ (controlled by agent $K$). Clearly, if any agent $I$ wishes to hide its objective functions, it needs to be smart in generating $R_{I,J}(x)$. 
\begin{assumption} \label{Asmp:ADMSETGROUP}
$(\mathcal{F},+)$ is an additive Abelian group.
\end{assumption}

For the purpose of discussion below, the admissible function set, $\mathcal{F}$ (Definition~\ref{Def:AdmFunSet}), is assumed to satisfy the following axioms. These axioms (A1-A5) are called Abelian axioms; and set $\mathcal{F}$ and operator $``+"$ that satisfies the Abelian axioms, forms an algebraic structure called Additive Abelian Group denoted by $(\mathcal{F},+)$. If the commutativity of the operator is not established, $(\mathcal{F},+)$ would simply be an additive group \cite{judson2010abstract} \cite{gallian2016contemporary}.
\begin{enumerate} [label=A\arabic*]
    \item $(\text{Additive Closure})$ The admissible function set $\mathcal{F}$ is closed under addition. 
    \item $(\text{Zero Element})$ There exists an element $0 \in \mathcal{F}$, such that, for any $f \in \mathcal{F}$, $f + 0 = f$.
    \item $(\text{Inverse Element})$ For every element $f \in \mathcal{F}$, there exists $-f \in \mathcal{F}$ such that, $f + (-f) = 0$.
    \item $(\text{Associativity})$ For all $f_1, f_2, f_3 \in \mathcal{F}$, $(f_1 + f_2) + f_3 = f_1 + (f_2 + f_3)$ holds.
    \item $(\text{Commutativity})$ For all $f_1, f_2 \in \mathcal{F}$, $f_1 + f_2 = f_2 + f_1$ holds.  
\end{enumerate}

\begin{example} [Bounded Degree Polynomials]
The set of all polynomials of degree less than or equal to $k$ is an admissible function set (if $p_m(x)$ denotes a $m$-degree polynomial in $x$, then $\mathcal{F} = \{p_m(x) = \sum_{i=0}^m a_i x^{i} \; | \; m \leq k\}$). Consider two polynomials in $x$ of degree $m$ and $l$ ($l < m \leq k$), denoted by $p_m(x) = \sum_{i=0}^m a_i x^{i}$ and $q_l(x) = \sum_{i=0}^l b_i x^{i}$ where $a_i, b_i \in \mathbb{R}$ for all $i$. Note, $p_m(x)+q_p(x) = \sum_{i=0}^{l} (a_i+b_i) x^{i} + \sum_{i=l+1}^m a_i x^{i}\in \mathcal{F}$ and the closure axiom A1 is satisfied. The zero element corresponding to $q_m$ is a polynomial with all coefficients ($a_i$) being zero. The inverse element corresponding to $p_m$ is a polynomial with $a_i = -b_i $. A4 and A5 are trivially satisfied for polynomials. Thus, $(\mathcal{F},+)$ is an additive Abelian group.
\end{example} 

\begin{example} [Even Degree Polynomials] 
The set of polynomials with even degrees upper bounded by $k$ is an admissible function set, $\mathcal{F} = \{p_k = \sum_{i=0}^{\lfloor k/2 \rfloor} a_{2i} x^{2i}\; | \; a_{2i} \in \mathbb{R} \}$. Additive closure, associativity and commutativity can be easily proved. The zero element is just an even degree polynomial with zero coefficients and inverse element is another polynomial with coefficients $b_{2i} = -a_{2i}$. 
\end{example}

Using the definition and closure properties of the admissible set $\mathcal{F}$ (due to $(\mathcal{F},+)$ being an additive Abelian group), we can clearly see that if functions $R_{I,J}(x) \in \mathcal{F}$, for all ($I$, $J$) pairs, then $\hat{f}_i(x) \in \mathcal{F}$. Another important consideration for selecting $R_{I,J}(x)$, comes from technical assumption \ref{Asmp:GradientCond}. The arbitrary functions are so selected that the obfuscated functions satisfy Assumptions~\ref{Asmp:GradientCond}, i.e. the obfuscated functions should have bounded gradients and must have Lipschitz gradients. It is clear from the above arguments that the strategy of generating functions $R_{I,J}(x)$ is problem specific and needs to be looked at on individually.

\section{Convergence Analysis and Privacy Results} \label{Sec:ConvPrivProof}
In this section, we prove the convergence of Algorithm~\ref{Algo:PrivDistOptNCFun}, present precise privacy claims and provide proofs.
\subsection{Convergence Analysis}
Each of the obfuscated objective functions are possibly non-convex functions with a convex aggregate ($f(x) = \sum_{i} f_i(x) = \sum_{i} \hat{f}_i(x)$ is convex, as noted in Remark~\ref{Rem:InvAgg}). Convergence results developed in \cite{gade16convsum} show that Algorithm~\ref{Algo:PrivDistOptNCFun} correctly minimizes convex aggregate of non-convex functions (and hence accurately solves the optimization problem in (\ref{Eq:OptProb})). 

\begin{theorem} [Theorem 5 and Claim 1, \cite{gade16convsum}]
Agent iterates ($x^I_k$) and the iterate average ($\bar{x}_k$) generated by Algorithm~\ref{Algo:PrivDistOptNCFun} (under Assumptions~\ref{Asmp:FunSet}, \ref{Asmp:GradientCond} and strong connectivity of the communication topology), asymptotically converge to an optimum in $\mathcal{X}^*$.
\end{theorem}
\begin{proof}
The proof follows as per the proof for Theorem 5 in \cite{gade16convsum}. \hfill $\blacksquare$
\end{proof}

\subsection{Privacy Guarantees}
We first establish a connectivity condition on the underlying communication topology. Note that we only need strong connectedness for convergence. The following conditions will be used to characterize graph topologies for which privacy guarantees can be provided. We begin with the definition of vertex connectivity and use it to define admissible topologies. 

\begin{definition} [Vertex Connectivity $\kappa(\mathcal{G})$, \cite{godsil2013algebraic}]  
The vertex connectivity of a connected graph $\mathcal{G}$ is the minimum number of vertices whose deletion would increase the number of connected components. 
\end{definition}

\begin{definition} [$f$-admissible topology] \label{Def:f-conntopo}
A graph $\mathcal{G}$ is called $f$-admissible if its vertex connectivity is greater than $f$, i.e. $\kappa(\mathcal{G}) > f$.
\end{definition}

We use a result from Whitney \cite{whitney1932congruent} that relates the vertex connectivity ($\kappa(\mathcal{G})$), the edge connectivity ($\lambda(\mathcal{G})$) and the minimum degree of the graph ($\delta(\mathcal{G})$). We will use this result to obtain a inequality on the minimum degree of a $f$-admissible graph.
\begin{theorem} [Whitney, \cite{whitney1932congruent}]
For any arbitrary graph $\mathcal{G}$, $\kappa(\mathcal{G}) \leq \lambda(\mathcal{G}) \leq \delta(\mathcal{G})$.
\end{theorem}
\noindent We know that for a $f$-admissible graph, $\kappa(\mathcal{G}) > f$, along with the above result we get, $\delta(\mathcal{G}) \geq \kappa(\mathcal{G}) > f \text{ implies } \delta(\mathcal{G}) \geq f + 1$. This condition is necessary to ensure the privacy of individual objective functions\footnote{We will characterize a few failure scenarios (loss of privacy) including the one with presence of nodes that have degree less than $f+1$. This is, however, not the only loss of privacy scenario.}. Topology being $f$-admissible is necessary and sufficient to ensure privacy of additive objective functions of type $f_\mathcal{I}(x) = \sum_{i \in \mathcal{I}} f_i(x)$, where $\mathcal{I} \subset \mathcal{V}-\mathcal{A}$.

We now formally state the privacy guarantee for our privacy preserving distributed optimization algorithm. As described in Section~\ref{Sec:AdvModel}, strong PC adversarial coalition has access to parameter vectors (from all agents) and the underlying graph topology. The problem is assumed to satisfy Assumptions~\ref{Asmp:FunSet} and \ref{Asmp:GradientCond} for correctness of the algorithm. 

\begin{theorem}
\label{Th:TPriv-2}
Let $\mathcal{A}$ denote a coalition of $f$ PC adversaries among $S$ agents. Let the underlying communication topology be $f$-admissible then Algorithm~\ref{Algo:PrivDistOptNCFun} is a privacy preserving algorithm in the sense of Definition~\ref{Def:Privacy}. 
\end{theorem}
\begin{proof}
We present a constructive method to show that given an execution (and corresponding observations), any guess of objective functions ($\in \mathcal{F}^S$, such that the sum of functions is $f(x)$) made by the coalition is equally likely.

We conservatively assume that the adversary can observe the obfuscated functions, $\hat{f}_i(x)$, the private objective functions of the coalition members $f_\mathscr{A}(x)$ ($\mathscr{A} \in \mathcal{A}$) and arbitrary functions transmitted from and received by each of the coalition members, $R_{\mathscr{A},J}$ and $R_{K,\mathscr{A}}$ ($J \in \mathcal{N}_{\mathscr{A}}$ and $K \text{ such that } \mathscr{A} \in \mathcal{N}_{K}$, for all $\mathscr{A} \in \mathcal{A}$). Since the adversaries also follow the same protocol (Algorithm~\ref{Algo:PrivDistOptNCFun}), the coalition is also aware of the fact that the private objective functions have been obfuscated by function sharing approach (Eq.~\ref{Eq:F_New}).

\begin{equation*}
\hat{f}_i(x) = f_i(x) + \sum_{K:I \in \mathcal{N}_K} R_{K,I}(x) - \sum_{J \in \mathcal{N}_I} R_{I,J}(x)
\end{equation*}

Clearly, one can rewrite this transformation approach, using signed incidence matrix of bidirectional graph $\mathcal{G}$ \cite{godsil2013algebraic} \cite{diestel2005graph}.
\begin{align}
\mathbf{\hat{f}} = \textbf{f} + \textbf{B} \textbf{R}.
\end{align}
where, $\mathbf{\hat{f}} = \begin{bmatrix}\hat{f}_1(x), \hat{f}_1(x), \ldots, \hat{f}_S(x)\end{bmatrix}^T$ is a $S \times 1$ vector of obfuscated functions $\hat{f}_i(x)$ for $i = \{1, 2, \ldots, S\}$, and $\textbf{f} = \begin{bmatrix} {f}_1(x), {f}_1(x), \ldots, {f}_S(x)\end{bmatrix}^T$ is a $S \times 1$ vector of private (true) objective functions, $f_i(x)$. $\textbf{B} = \begin{bmatrix}
B_C, -B_C
\end{bmatrix}$, where $B_C$ (of dimension $S \times |\mathcal{E}|/2$) is the incidence matrix of a directed graph obtained by considering only one of the directions of every bidirectional edge in graph $\mathcal{G}$\footnote{This represents an orientation of graph $\mathcal{G}$ \cite{godsil2013algebraic}.}. Each column of $\mathbf{B}$ represents a directed communication link between any two agents. Hence, any bidirectional edge between agents $I$ and $J$ is represented as two directed links, $I$ to $J$, $(I,J) \in \mathcal{E}$ and $J$ to $I$ $(J,I) \in \mathcal{E}$ and corresponds to two columns in $\mathbf{B}$. $\textbf{R}$ represents a $|\mathcal{E}| \times 1$ vector consisting of functions $R_{I,J}(x)$. Each entry in vector $\mathbf{R}$, function $R_{I,J}(x)$ corresponds to a column of $\mathbf{B}$ which, in turn corresponds to link $(I,J) \in \mathcal{E}$; and similarly, function $R_{J,I}(x)$ corresponds to a different column of $\mathbf{B}$ which, in turn corresponds to link $(J,I) \in \mathcal{E}$. Note that $\ell^\text{th}$ row of column vector $\mathbf{R}$ corresponds to $\ell^\text{th}$ column of incidence matrix $\mathbf{B}$. 

We will show that, two different sets of true objective functions ($\mathbf{f}$ and $\mathbf{f}^o$) and correspondingly two different set of arbitrary functions ($\mathbf{R}$ and $\mathbf{G}$), can lead to exactly same execution and observations for the adversaries\footnote{$\textbf{f}$ and $\textbf{f}^o$ are dissimilar and arbitrarily different.}. We want to show that both these cases can result in same obfuscated objective functions. That is,
\begin{align}
\mathbf{\hat{f}} = 
\textbf{f} + 
\textbf{B} \textbf{R} = 
\textbf{f}^o + 
\textbf{B} \textbf{G}.
\label{Eq:ProofEq1}
\end{align}

We will show that given any set of private objective functions $\mathbf{f}^o \in \mathcal{F}^S$, suitably selecting arbitrary functions $G_{I,J}(x)$ corresponding to links incident at ``good" agents, it is possible to make $\mathbf{f}^o$ indistinguishable from original private objective functions $\mathbf{f}$, solely based on the execution observed by the adversaries. We do so by determining entries of $\textbf{G}$, which are arbitrary functions that are dissimilar from $R_{I,J}(x)$ when $I$ and $J$ are both ``good". The design $\mathbf{G}$ such that the obfuscated objective functions $\mathbf{\hat{f}}$ are the same for both situations. 

Since adversaries observe arbitrary functions corresponding to edges incident to and from them, we set the arbitrary functions corresponding to edges incident on adversaries as $G_{K,\mathscr{A}} = R_{K,\mathscr{A}}$ and arbitrary functions corresponding to edges incident away the adversary as $G_{\mathscr{A},J} = R_{\mathscr{A},J}$ (where $K: \mathscr{A} \in \mathcal{N}_{K}$ and $J \in \mathcal{N}_{\mathscr{A}}$, for all $\mathscr{A} \in \mathcal{A}$). Now, we define $\mathbf{\tilde{G}}$ as the vector containing all elements of $\mathbf{G}$ except those corresponding to the edges incident to and from the adversaries\footnote{The only entries of $\mathbf{G}$, that are undecided at this stage are included in $\mathbf{\tilde{G}}$. These are functions $G_{I,J}$ such that $I$, $J$ are both ``good".}. Similarly, we define $\mathbf{\tilde{B}}$ to be the new incidence matrix obtained after deleting all edges that are incident on the adversaries (i.e. deleting columns corresponding to the links incident on adversaries, from the old incidence matrix $\mathbf{B}$). We subtract $G_{\mathscr{A},J}(x)$ and $G_{K,\mathscr{A}}(x)$ ($\forall \ \mathscr{A} \in \mathcal{A}$) by subtracting them from $\mathbf{[\mathbf{\hat{f}} - \textbf{f}^o]}$ (in Eq.~\ref{Eq:ProofEq1}) to get effective function difference denoted by $\mathbf{[\mathbf{\hat{f}} - \textbf{f}^o]}_{\rm eff}$ as follows,
\begin{align}
\mathbf{[\mathbf{\hat{f}} - \textbf{f}^o]} &= \mathbf{B} \mathbf{G} = \mathbf{f} - \mathbf{f}^o +\mathbf{B} \mathbf{R},  && \ldots (\text{From Eq.~\ref{Eq:ProofEq1}})\\
\mathbf{[\mathbf{\hat{f}} - \textbf{f}^o]}_{\rm eff} &= [\mathbf{\hat{f}} - \mathbf{f}^o] - \sum_{\mathscr{A} \in \mathcal{A}} \left[ \sum_{K:\mathscr{A} \in \mathcal{N}_K} G_{K,\mathscr{A}}(x) - \sum_{J\in \mathcal{N}_\mathscr{A}} G_{\mathscr{A},J}(x) \right] = \mathbf{\tilde{B}} \mathbf{\tilde{G}}, \label{Eq:ProofEq2}
\end{align}
where, if $d$ entries of $\mathbf{G}$ were fixed\footnote{Total number of edges incident to and from adversaries is $d$. We fixed them to be the same as corresponding entries from $\mathbf{R}$, since the coalition can observe them.} then $\mathbf{\tilde{G}}$ is a $(|\mathcal{E}|-d) \times 1$ vector and $\mathbf{\tilde{B}}$ is a matrix with dimension $S \times (|\mathcal{E}| - d)$. The columns deleted from $\mathbf{B}$ correspond to the edges that are incident to and from the adversaries. Hence, $\mathbf{\tilde{B}}$ represents the incidence of a graph with these edges deleted.

We know from the $f$-admissibility of the graph, that $\mathbf{\tilde{B}}$ connects all the non-adversarial agents into a connected component\footnote{The adversarial nodes become disconnected due to the deletion of edges incident on adversaries (previous step).}. Since, the remaining edges form a connected component, the edges can be split into two groups. A group with edges that form a spanning tree over the good nodes (agents) and all other edges in the other group (see Remark~\ref{Ex:STEE} and Figure~\ref{Fig:ExSTEE}). Let $\mathbf{\tilde{B}}_{\rm ST}$ represent the incidence matrix\footnote{Its columns correspond to the edges that form spanning tree.} of the spanning tree and $\mathbf{\tilde{G}}_{\rm ST}$ represents the arbitrary functions corresponding to the edges of the spanning tree. $\mathbf{\tilde{B}}_{\rm EE}$ represents the incidence matrix formed by all other edges and $\mathbf{\tilde{G}}_{\rm EE}$ represents the arbitrary functions related to all other edges.
\begin{align}
\mathbf{[\mathbf{\hat{f}} - \textbf{f}^o]}_{\rm eff} = \begin{bmatrix}
\mathbf{\tilde{B}}_{\rm ST} & \mathbf{\tilde{B}}_{\rm EE}
\end{bmatrix} \begin{bmatrix}
\mathbf{\tilde{G}}_{\rm ST} \\
\mathbf{\tilde{G}}_{\rm EE}
\end{bmatrix}
=\mathbf{\tilde{B}}_{\rm ST} \mathbf{\tilde{G}}_{\rm ST} + \mathbf{\tilde{B}}_{\rm EE} \mathbf{\tilde{G}}_{\rm EE}.
\end{align}

\noindent We now arbitrary assign functions to elements of $\mathbf{\tilde{G}}_{\rm EE}$ (as per Section~\ref{Sec:FunDetails}) and then compute the arbitrary weights for $\mathbf{\tilde{G}}_{\rm ST}$. We know that the columns of $\mathbf{\tilde{B}}_{\rm ST}$ are linearly independent, since $\mathbf{\tilde{B}}_{\rm ST}$ is the incidence matrix of a spanning tree (cf. Lemma~2.5 in \cite{bapat2010graphs}). Hence, the left pseudoinverse\footnote{$A^\dagger$ represents the pseudoinverse of matrix $A$.} of $\mathbf{\tilde{B}}_{\rm ST}$ exists; and $\mathbf{\tilde{B}}_{\rm ST}^ \dagger \mathbf{\tilde{B}}_{\rm ST} = \mathbb{I}$, giving us the solution for $\mathbf{\tilde{G}}_{\rm ST}$ \footnote{An alternate way to look at this would be to see that $\mathbf{\tilde{B}}_{\rm ST}^T \mathbf{\tilde{B}}_{\rm ST}$ represents the edge Laplacian \cite{zelazo2007agreement} of the spanning tree. The edge Laplacian of an acyclic graph is non-singular and this also proves that left-pseudoinverse of $\mathbf{\tilde{B}}_{\rm ST}$ exists.}.
\begin{align}
\mathbf{\tilde{G}}_{\rm ST} = \mathbf{\tilde{B}}_{\rm ST} ^ \dagger \left[ \mathbf{[\mathbf{\hat{f}} - \textbf{f}^o]}_{\rm eff} - \mathbf{\tilde{B}}_{\rm EE} \mathbf{\tilde{G}}_{\rm EE} \right]. \label{Eq:SolutionEq}
\end{align}

Using the construction shown above, for any $\textbf{f}^o$ we can construct $\textbf{G}$ such that the execution as seen by adversaries is exactly the same as the original problem where the objective is $\textbf{f}$ and the arbitrary functions are $\textbf{R}$. A strong PC coalition cannot distinguish between two executions involving $\textbf{f}^o$ and $\textbf{f}$. Hence, no coalition can estimate $f_i(x)$ ($i \cancel{\in} \mathcal{A}$). $\hfill \blacksquare$
\end{proof}

\begin{figure}[!t]
\begin{subfigure}{.48\textwidth}
  \centering
  \includegraphics[width=0.75\linewidth]{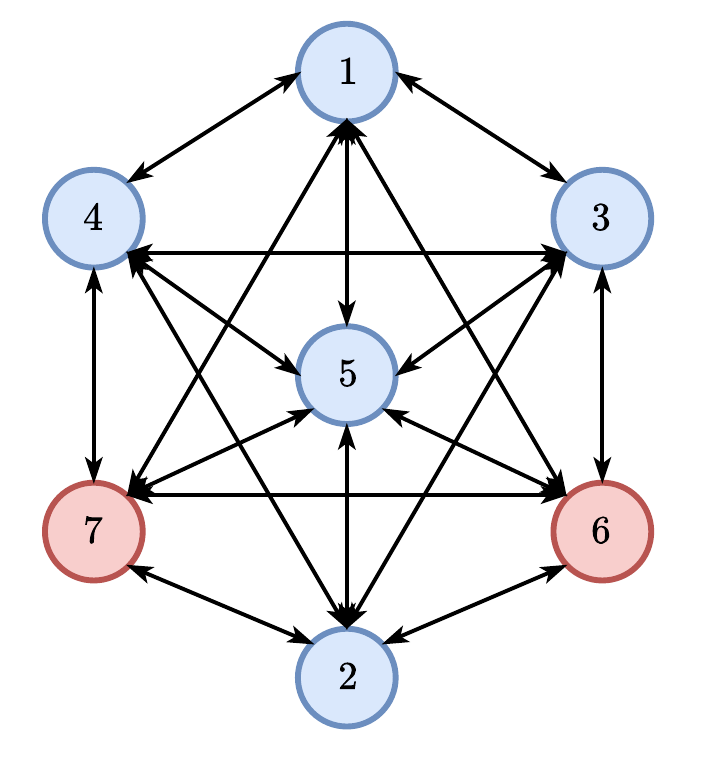}
  \caption{Example }
  \label{Fig:ExSTEE1}
\end{subfigure} \hfill
\begin{subfigure}{.48\textwidth}
  \centering
  \includegraphics[width=0.77\linewidth]{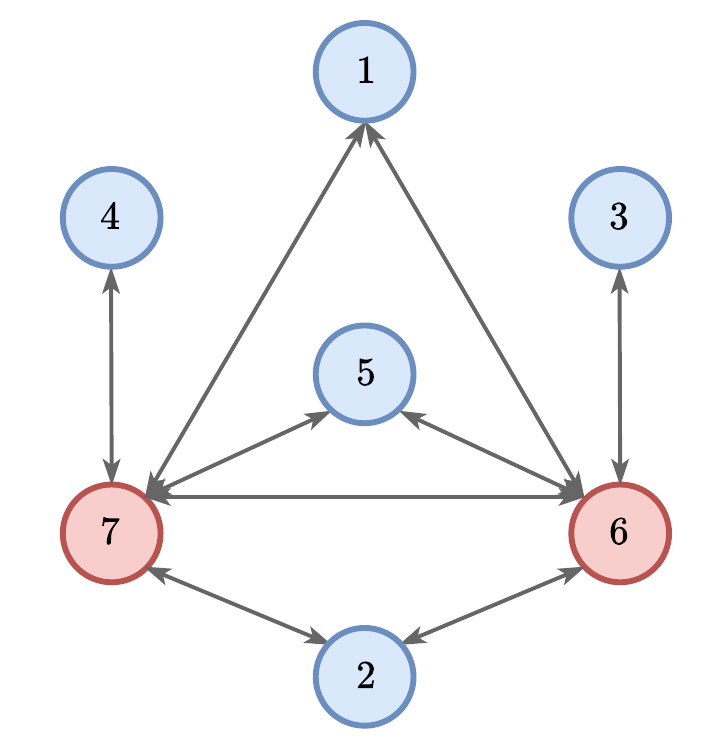}
  \caption{Step-1. Fix $\mathbf{G}_{\mathscr{A},l}$ and $\mathbf{G}_{l,\mathscr{A}}$.}
  \label{Fig:ExSTEE2}
\end{subfigure} \\
\begin{subfigure}{.48\textwidth}
  \centering
  \includegraphics[width=0.75\linewidth]{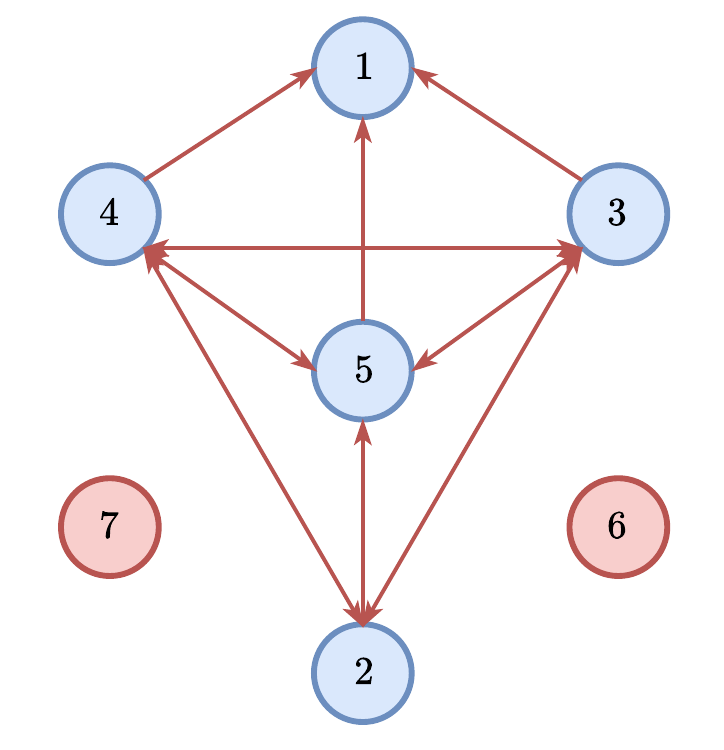}
  \caption{Step-2. Arbitrarily select $\mathbf{G}_{\rm EE}$. }
  \label{Fig:ExSTEE3}
\end{subfigure} \hfill
\begin{subfigure}{.48\textwidth}
  \centering
  \includegraphics[width=0.77\linewidth]{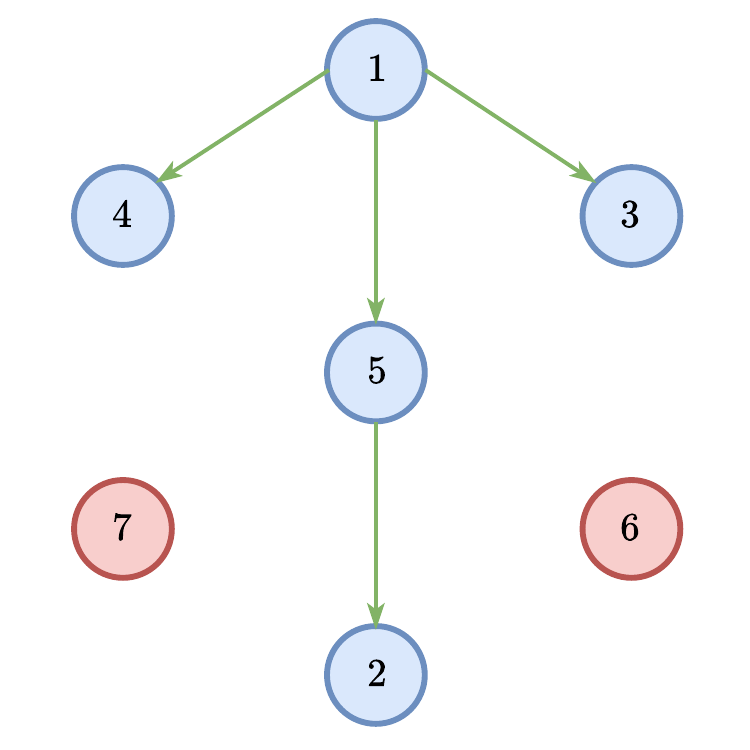}
  \caption{Step-3. Solve for $\mathbf{G}_{\rm ST}$ (Eq.~\ref{Eq:SolutionEq}).}
  \label{Fig:ExSTEE4}
\end{subfigure}
\caption{An example of construction used for proving Theorem~\ref{Th:TPriv-2}. Network has $S=7$ nodes and a coalition with $f=2$ adversaries. The graph topology is $2$-admissible.}
\label{Fig:ExSTEE}
\end{figure}

\begin{remark} [Method for Constructing $\mathbf{G}$]
\label{Ex:STEE}
We present an example for the construction used in the above proof. Let us consider a system of $S=7$ agents communicating under a $2$-admissible topology (see Figure~\ref{Fig:ExSTEE1}). A coalition of two PC adversaries ($\mathcal{A} = \{6, 7\}$, $f=2$) is a part of the system. We can divide the task of constructing $\mathbf{G}$ into three steps - 
\begin{enumerate}
    \item Fix $\mathbf{G}_{\mathscr{A},l}$ and $\mathbf{G}_{l,\mathscr{A}}$ (links incident on adversaries) to be the corresponding entries in $\mathbf{R}$,
    \item Arbitrarily select the functions corresponding to non spanning tree edges ($\tilde{\mathbf{G}}_{\rm EE}$), and
    \item Solve for the functions corresponding to the spanning tree ($\tilde{\mathbf{G}}_{\rm ST}$) using Eq.~\ref{Eq:SolutionEq}.
\end{enumerate}

We first, follow the Step 1 and fix $G_{K,\mathscr{A}} = R_{K,\mathscr{A}}$ and $G_{\mathscr{A},J} = R_{\mathscr{A},J}$ (where $K: \mathscr{A} \in \mathcal{N}_{K}$ and $J \in \mathcal{N}_{\mathscr{A}}$, for all $\mathscr{A} \in \mathcal{A}$). Step 1 follows form the fact that the adversaries observe $R_{\mathscr{A},J}$ and $R_{K,\mathscr{A}}$, and hence they need to be same in both executions. This is followed by substituting the known entries in $\mathbf{G}$ and subtract them from the left hand side as shown in Eq.~\ref{Eq:ProofEq2}. This corresponds to the deletion of all incoming and outgoing edges from the adversaries. 
The incidence matrix of this new graph is denoted by $\mathbf{\tilde{B}}$. The edges in the new graph can be decomposed into two groups - a set containing edges that form a spanning tree and a set that contains all other edges. This is seen in Figure~\ref{Fig:ExSTEE3} where the red edges are all the remaining links (incidence matrix, $\mathbf{\tilde{B}}_{\rm EE}$); and Figure~\ref{Fig:ExSTEE4} where the green edges form a spanning tree (incidence matrix, $\mathbf{\tilde{B}}_{\rm ST}$) with Agent 1 as the root and all other ``good" agents as its leaves (agents 2, 3, 4, 5). We now arbitrarily select $\tilde{\mathbf{G}}_{\rm EE}$ (Step-2) followed by solving for $\tilde{\mathbf{G}}_{\rm ST}$ using Eq.~\ref{Eq:SolutionEq} (Step-3). This completes construction of $\mathbf{G}$.
\end{remark}

\begin{example}
The privacy preserving optimization algorithm completely obfuscates the true objective functions $f_i(x)$ (topology as shown in Figure~\ref{Fig:F-0}). We can have, for arbitrarily large number of problems, executions that are exactly same. For instance consider three agents with one adversary (Agent 1). We show that for two different sets of objective functions $\{ f_1(x),f_2(x),f_3(x) \}$ and $\{h_1(x), h_2(x), h_3(x)\}$ \footnote{The two problems are restricted to have, $\sum_{i=1}^3 f_i(x) = \sum_{i=1}^3 h_i(x) = f(x)$.}, there exists two sets of arbitrarily functions, $R_{I,J}(x)$, such that $\hat{f}_1(x) = \hat{h}_1(x), \hat{f}_2(x) = \hat{h}_2(x)$ and $\hat{f}_3(x) = \hat{h}_3(x)$, while the adversaries observe the same execution. Column 1 of Table~\ref{Tab:2Exec} shows the objective functions for two distributed optimization problems. Now we consider two sets of arbitrarily functions, $R_{I,J}(x)$, as seen in Column 2. Observe that the obfuscated objective functions for both problems are exactly the same (Column 3), giving rise to identical executions (iterations) and the quantities observed by the adversary ($f_1(x), \hat{f}_1(x), \hat{f}_2(x), \hat{f}_3(x), R_{1,2}(x), R_{1,3}(x), R_{2,1}(x)$ and $R_{3,1}(x)$) are the same for both problems. This results in the inability of the adversary to distinguish between the two different optimization problems. 
\begin{table}[!h]
\centering
\begin{tabular}{ |l|l|l| }
\hline
\multicolumn{3}{ |c| }{\textbf{Problem 1}} \\
\hline
$\bm{f_i(x)}$ & $\bm{R_{I,J}(x)}$ & $\bm{\hat{f}_i(x)}$ \\ \hline
\multirow{2}{*}{$f_1(x) = x^2$} & $R_{1,2}(x) = 3x + 9 x^2 + x^3 + 2x^4$ &  \multirow{2}{*}{$\hat{f}_1(x) = -3x -4x^2 -4x^3 + 2x^4$}  \\
 & $R_{1,3}(x) = 5x + 1x^2 + 7x^3 + 6x^4$& \\ \hline
\multirow{2}{*}{$f_2(x) = x^2 + x^4$} & $R_{2,1}(x) = 5x^2+ 3x^3 + 6x^4$ &  \multirow{2}{*}{$\hat{f}_2(x) = 10x + 4x^2 -7x^3 -4x^4$}  \\
 & $R_{2,3}(x) = 4x^2 + 5x^3 + 7x^4$& \\ \hline
 \multirow{2}{*}{$f_3(x) = x^4$} & $R_{3,1}(x) = 5x + 1x^3 + 4x^4$ &  \multirow{2}{*}{$\hat{f}_3(x) = -7x + 2x^2 + 11x^3 + 4x^4$}  \\
 & $R_{3,2}(x) = 7x + 3x^2 + 6x^4 $& \\ \hline 
 \hline
\multicolumn{3}{ |c| }{\textbf{Problem 2}} \\
\hline
$\bm{f_i(x)}$ & $\bm{R_{I,J}(x)}$ & $\bm{\hat{f}_i(x)}$ \\ \hline
\multirow{2}{*}{$h_1(x) = x^2$} & $R_{1,2}(x) = 3x + 9 x^2 + x^3 + 2x^4$ &  \multirow{2}{*}{$\hat{h}_1(x) = -3x -4x^2 -4x^3 + 2x^4$}  \\
 & $R_{1,3}(x) = 5x + 1x^2 + 7x^3 + 6x^4$& \\ \hline
\multirow{2}{*}{$h_2(x) = 3x^2 + 3x^4$} & $R_{2,1}(x) = 5x^2+ 3x^3 + 6x^4$ &  \multirow{2}{*}{$\hat{h}_2(x) = 10x + 4x^2 -7x^3 -4x^4$}  \\
 & $R_{2,3}(x) = -17x +10x^2 +12 x^3 + 13x^4$& \\ \hline
\multirow{2}{*}{$h_3(x) = -2x^2 -x^4$} & $R_{3,1}(x) = 5x + 1x^3 + 4x^4$ &  \multirow{2}{*}{$\hat{h}_3(x) = -7x + 2x^2 + 11x^3 + 4x^4$}  \\
 & $R_{3,2}(x) = -10x + 7x^2 + 7x^3 + 10x^4 $& \\\hline
\end{tabular}
\caption{Table describes two distributed optimization problems that will have identical executions.}
\label{Tab:2Exec}
\end{table}

\end{example} 

Theorem~\ref{Th:TPriv-2} states that none of the private objective functions can be meaningfully recovered (Definition~\ref{Def:Privacy}). However, we can further talk about the privacy of the overall aggregate function, $f(x)$, and the aggregate of objective functions belonging to ``good" agents, $f_{\mathcal{A}'}$, under relaxed requirements (like allowing ambiguity in the constant term). We also allow the coalition to run numerical schemes like polynomial regression on observed quantities. With these added assumptions we can claim that a coalition may uncover (not necessarily) the overall aggregate and the aggregate of objective functions private to ``good" agents. 

\begin{remark} [Privacy of $f(x)$] \label{Rem:PrivF(x)}
The invariance of aggregate is described in Remark~\ref{Rem:InvAgg}. Clearly, if the coalition can extract $\hat{f}_i(x)$ via a numerical scheme (like polynomial regression) albeit with ambiguity in the constant term, then by simply adding all $\hat{f}_i(x)$ together, the coalition can estimate $f(x)$ with uncertainty in the constant term.
\end{remark}

\begin{remark} [Privacy of $f_{\mathcal{A}'}(x)$]
We know from Remark~\ref{Rem:PrivF(x)} that the coalition can estimate $f(x)$ and the coalition already has access to the objective functions of the adversaries. Hence, the coalition can estimate $f_{\mathcal{A}'}(x)$ using the relation $f_{\mathcal{A}'}(x) = \sum_{i \cancel{\in} \mathcal{A}} f_i(x) = f(x) - \sum_{i \in \mathcal{A}} f_i(x)$. The coalition can estimate $f_{\mathcal{A}'}(x)$ with ambiguity in constant term.
\end{remark}

\subsection{Loss of Privacy}
Theorem~\ref{Th:TPriv-2} clearly characterizes the conditions on communication topology for privacy preservation. The $f$-admissibility condition is both necessary and sufficient for privacy of Algorithm~\ref{Algo:PrivDistOptNCFun}. In this section we characterize certain topologies that (do not satisfy $f$-admissibility condition) and are prone to specific types of attacks.
\begin{itemize}
    \item Loss of individual privacy - If the degree of one or more nodes is less than $f+1$, then those agents are susceptible to privacy loss and a coalition may uncover their objective functions. Figure~\ref{Fig:FailScene1} shows a system of five agents ($S=5$) with one adversary ($\mathcal{A} = \{3\}$, $f = 1$) and its communication topology is not $1$-admissible. Agent 4 has a degree deficiency. Since the adversary can observe all the arbitrary functions ($R_{3,4}$, $R_{4,3}$) used by Agent 4 for obfuscating its objective function, the objective function of Agent 4 ($f_4(x)$) is vulnerable to privacy loss from the adversary.  
    \item Loss of group privacy - If the vertex is only $(f-1)$-admissible and the coalition is a of size $f$, then there exists cases where the coalition can effectively uncover sum of objective functions of a set of agents (group). Figure~\ref{Fig:FailScene2} shows a system with six agents ($S=6$) with one adversary ($\mathcal{A}= \{3\}$, $f = 1$) under a $0$-admissible topology. The objective function sums, $f_{\{1,2\}}(x) = f_1(x) + f_2(x)$ and $f_{\{4,5,6\}}(x) =  f_4(x)+f_5(x)+f_6(x)$ are susceptible to privacy loss from the adversary.
\end{itemize}

\begin{figure}[h]
\begin{subfigure}{.45\textwidth}
  \centering
  \includegraphics[width=0.93\linewidth]{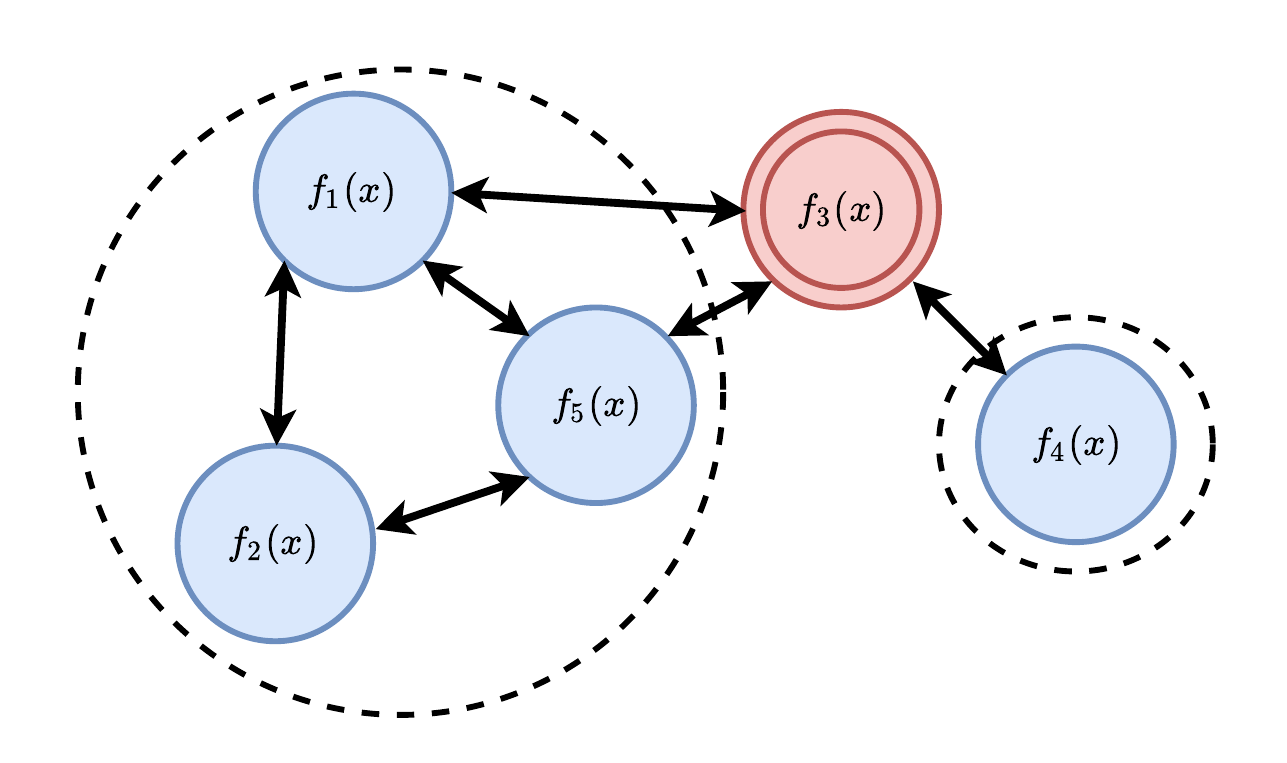}
  \caption{Loss of individual privacy.}
  \label{Fig:FailScene1}
\end{subfigure} \hfill
\begin{subfigure}{.5\textwidth}
  \centering
  \includegraphics[width=1.05\linewidth]{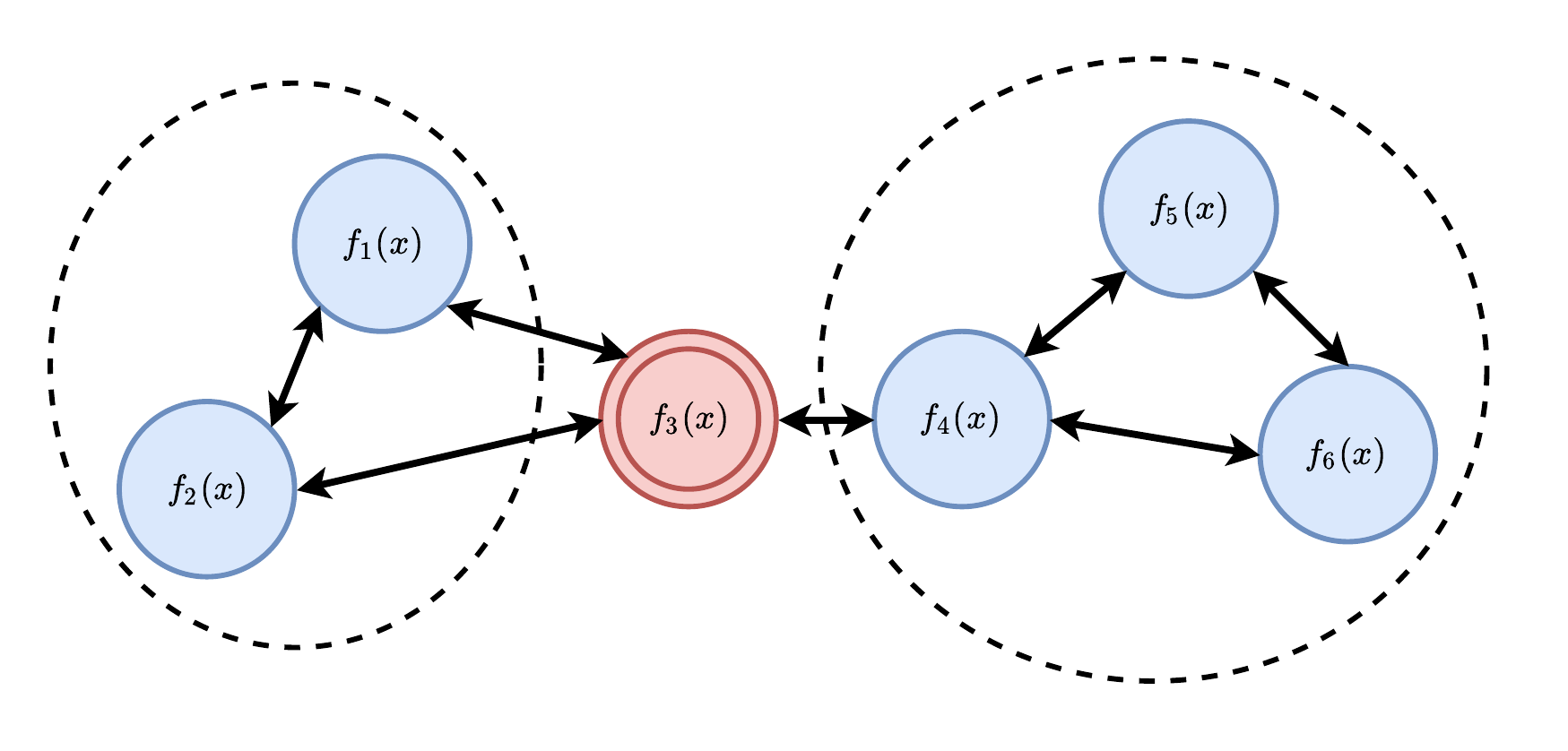}
  \caption{Loss of group privacy.}
  \label{Fig:FailScene2}
\end{subfigure} 
\caption{Loss of privacy - failure scenarios. The adversary is shown in red with concentric circles.}
\label{Fig:FailScene}
\end{figure}

\section{Simulation Results} \label{Sec:Results}

In this section, we show that the privacy preserving optimization algorithm in Section~\ref{Sec:Algo} solves the learning problem correctly. We consider a network of $S = 3$ connected agents as shown in Figure~\ref{Fig:F-0}. Each agent is endowed with a private objective function, given by $f_1(x) = x^2$, $f_2(x) = x^2 + x^4$ and $f_3(x) = x^4$. The system tries to optimize the additive aggregate, $f(x) = f_1(x) + f_2(x) + f_3(x) = 2(x^2 + x^4)$. The aggregate function $f(x)$ achieves optimum (minimum) at $x^* = 0$. Agents implement the privacy preserving distributed optimization protocol presented as Algorithm~\ref{Algo:PrivDistOptNCFun}. The agents share their states (model parameters) with neighbors. States received from other agents are fused using a doubly-stochastic matrix $B_k$. The learning step is given by a monotonically decreasing, non-summable but square summable sequence, $\alpha_k = 1/(k + 0.0001)$.
$$B_k = \begin{bmatrix}
0.5 & 0.25 & 0.25 \\
0.25 & 0.5 & 0.25 \\
0.25 & 0.25 & 0.5
\end{bmatrix}.$$

\begin{figure}[!h]
\begin{subfigure}{.49\textwidth}
  \centering
  \includegraphics[width=1.1\linewidth]{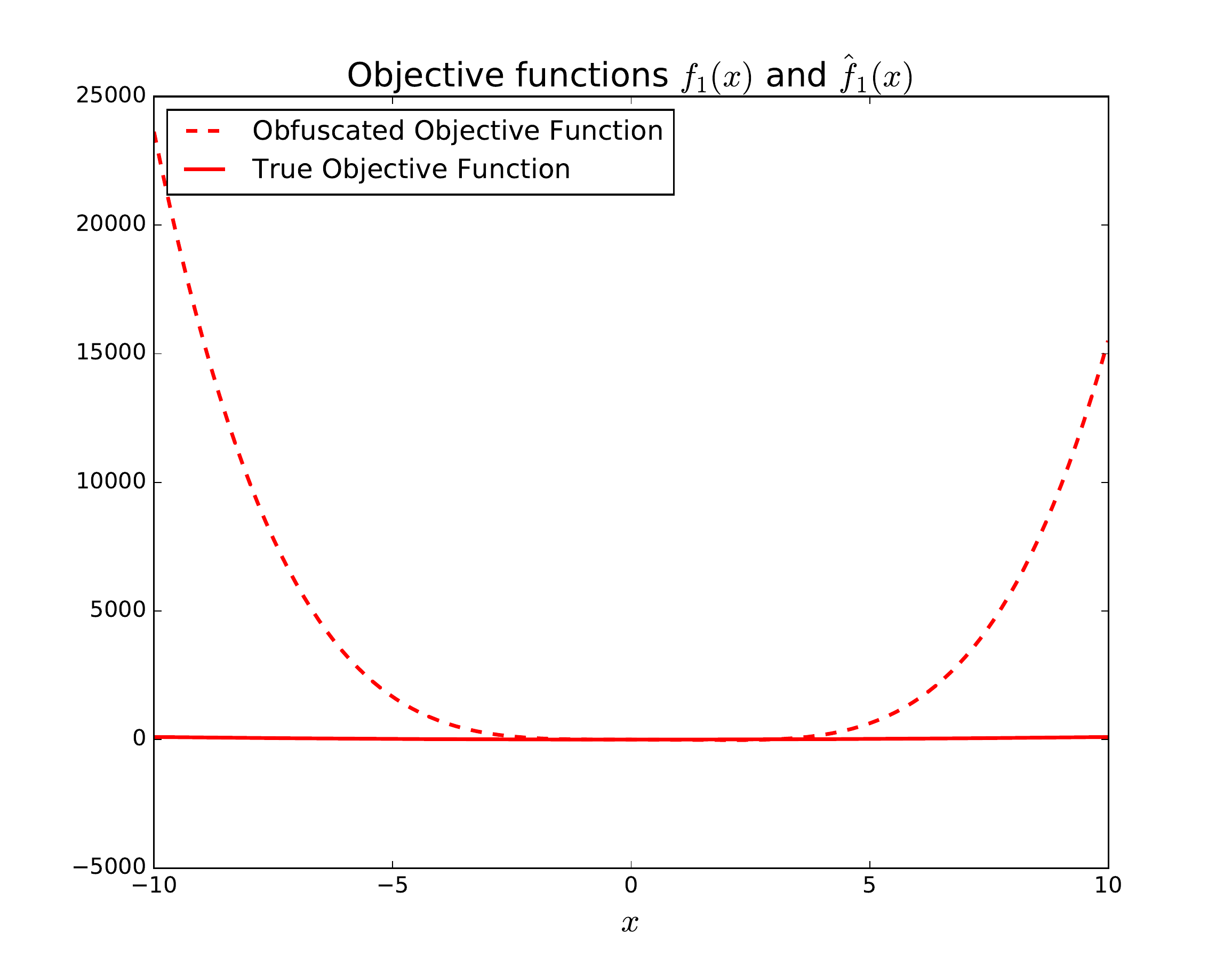}
  \caption{$f_1(x)$ and $\hat{f}_1(x)$}
  \label{Fig:Sim7}
\end{subfigure} \hfill
\begin{subfigure}{.49\textwidth}
  \centering
  \includegraphics[width=1.1\linewidth]{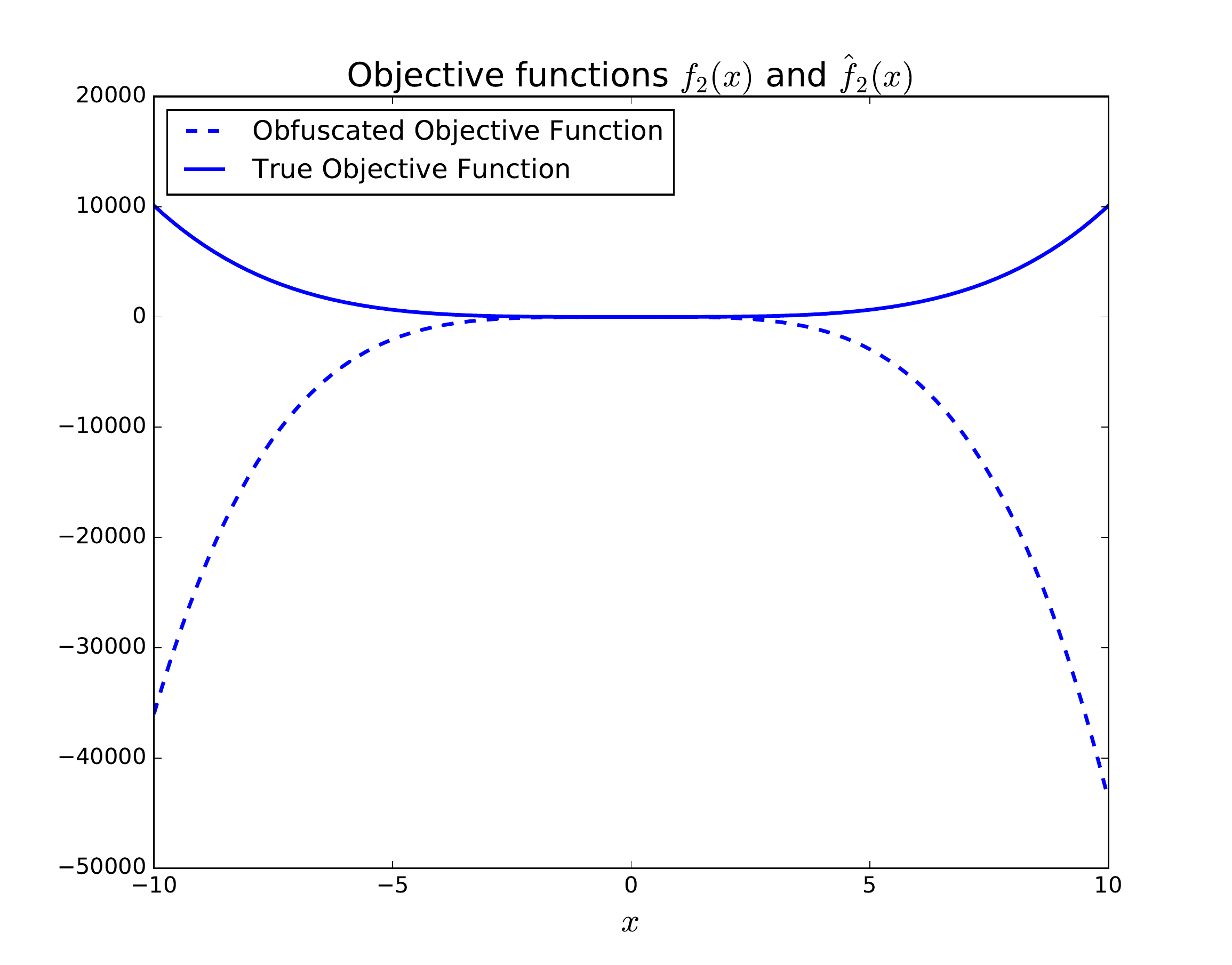}
  \caption{$f_2(x)$ and $\hat{f}_2(x)$}
  \label{Fig:Sim8}
\end{subfigure} \\
\centering
\begin{subfigure}{.49\textwidth}
  \centering
  \includegraphics[width=1.1\linewidth]{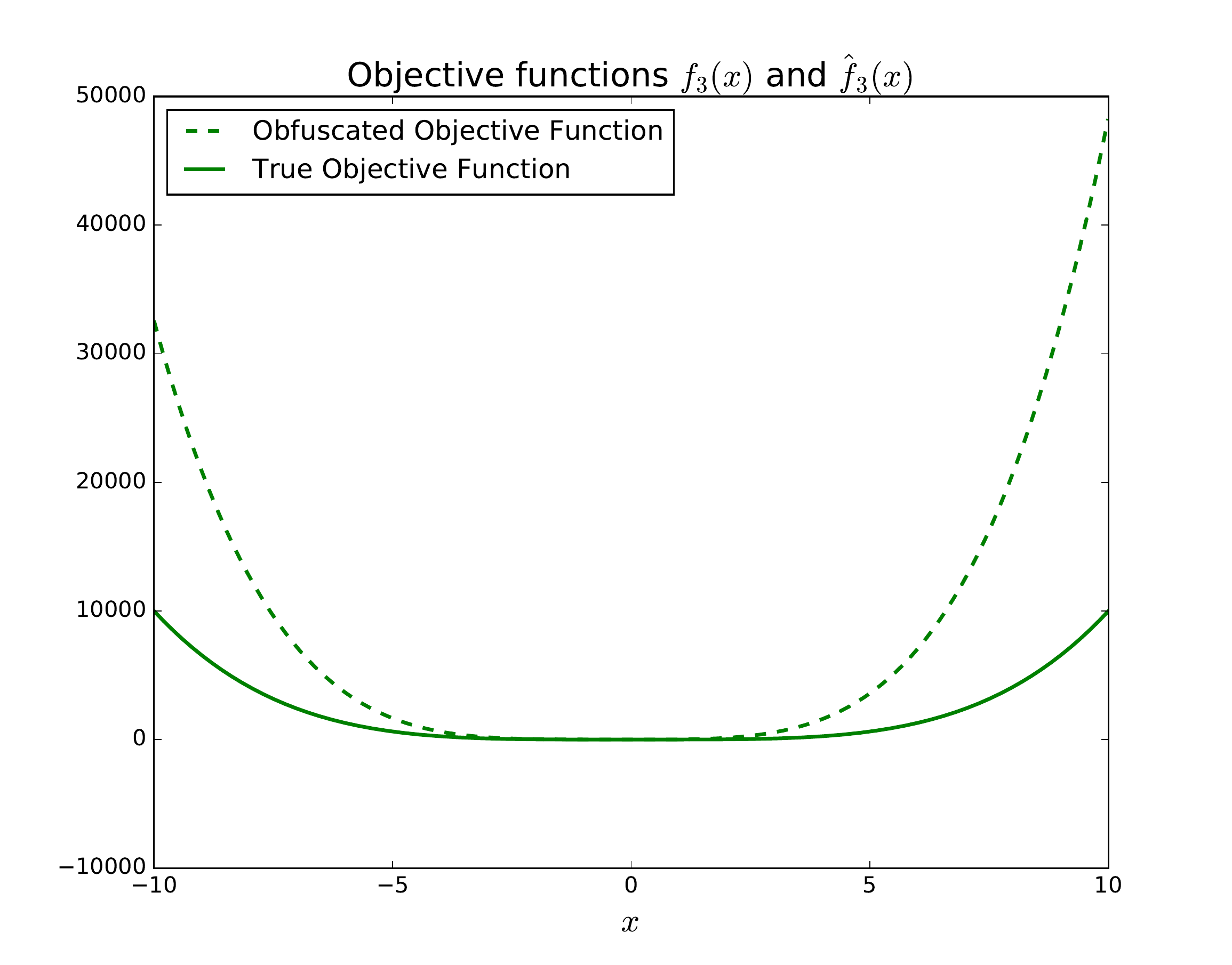}
  \caption{$f_3(x)$ and $\hat{f}_3(x)$}
  \label{Fig:Sim9}
\end{subfigure} 
\caption{Objective Function and Obfuscated Objective Functions.}
\label{Fig:Simulation3}
\end{figure}

We consider arbitrarily generated function $R_{I,J}(x)$ in Table~\ref{Tab:2Exec} (Problem 1). The obfuscated objective functions are given by (Column 3), 
\begin{align*}
\hat{f}_1(x) &= f_1(x) + R_{2,1}(x) + R_{3,1}(x) - R_{1,2}(x) - R_{1,3}(x) = -3x -4x^2 -4x^3 + 2x^4, \\
\hat{f}_2(x) &= f_2(x) + R_{1,2}(x) + R_{3,2}(x) - R_{2,1}(x) - R_{2,3}(x) = 10x + 1x^2 -4x^3 -4x^4, \text{ and } \\
\hat{f}_3(x) &= f_3(x) + R_{1,3}(x) + R_{2,3}(x) - R_{3,1}(x) - R_{3,2}(x) = -7x + 5x^2+8x^3+4x^4.
\end{align*}

It can be easily verified that the aggregate function is an invariant (Remark~\ref{Rem:InvAgg}), $f(x) = \sum f_i(x) = \sum \hat{f}_i(x) = 2(x^2 + x^4)$. Note, that transformed problem falls under the ``distributed optimization of a convex aggregate of non-convex functions" framework, presented in \cite{gade16convsum}. Figure~\ref{Fig:Sim7}, \ref{Fig:Sim8}, and \ref{Fig:Sim9} shows the effect of obfuscation using arbitrary random functions $R_{I,J}(x)$ on the private objective function $f_i(x)$. As seen in Figure~\ref{Fig:Simulation3}, the private objective function and obfuscated functions are different from each other. 

\begin{figure}[!h]
\begin{subfigure}{.48\textwidth}
  \centering
  \includegraphics[width=1.1\linewidth]{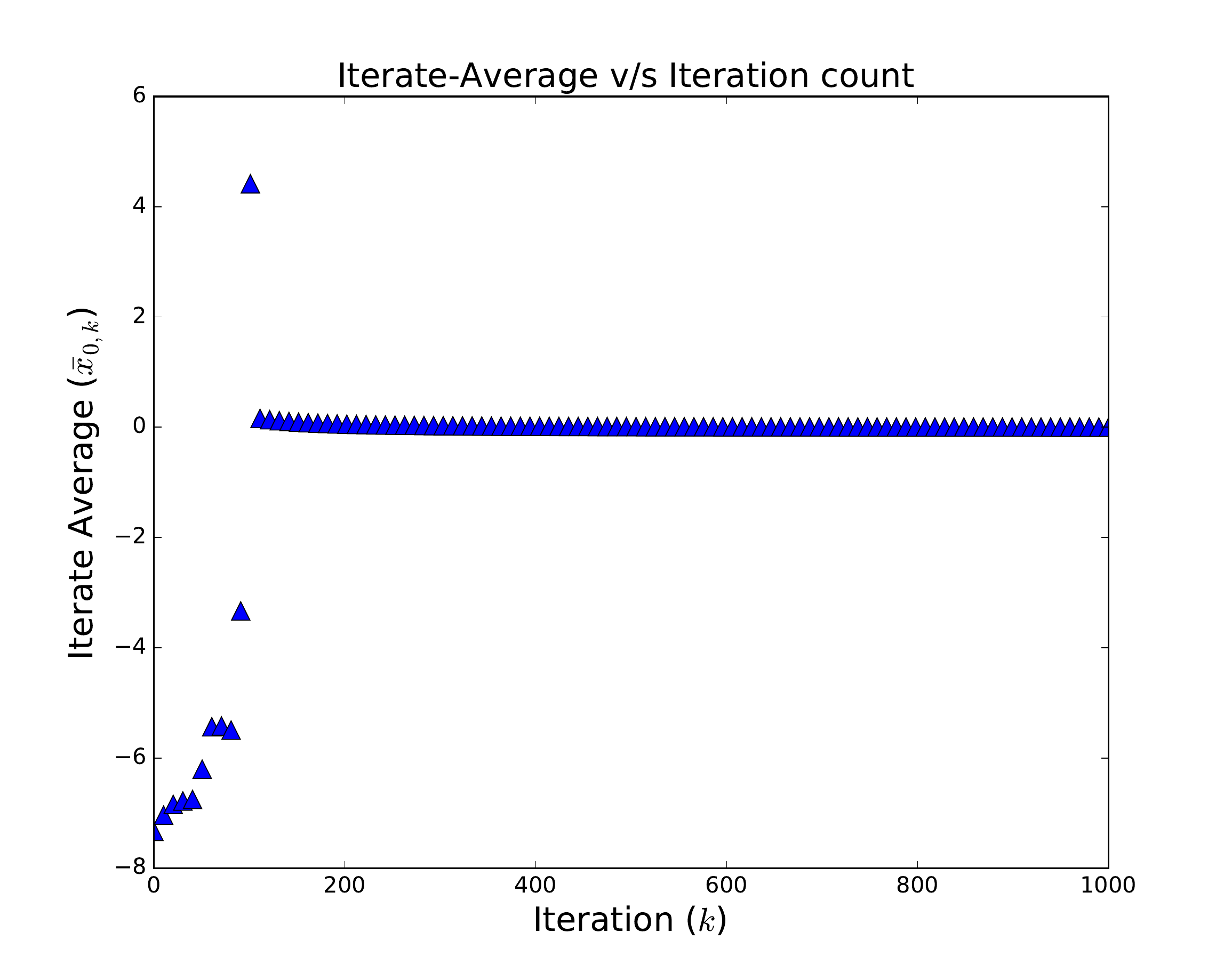}
  \caption{Convergence of iterate-average.}
  \label{Fig:Sim1}
\end{subfigure} \hfill
\begin{subfigure}{.48\textwidth}
  \centering
  \includegraphics[width=1.1\linewidth]{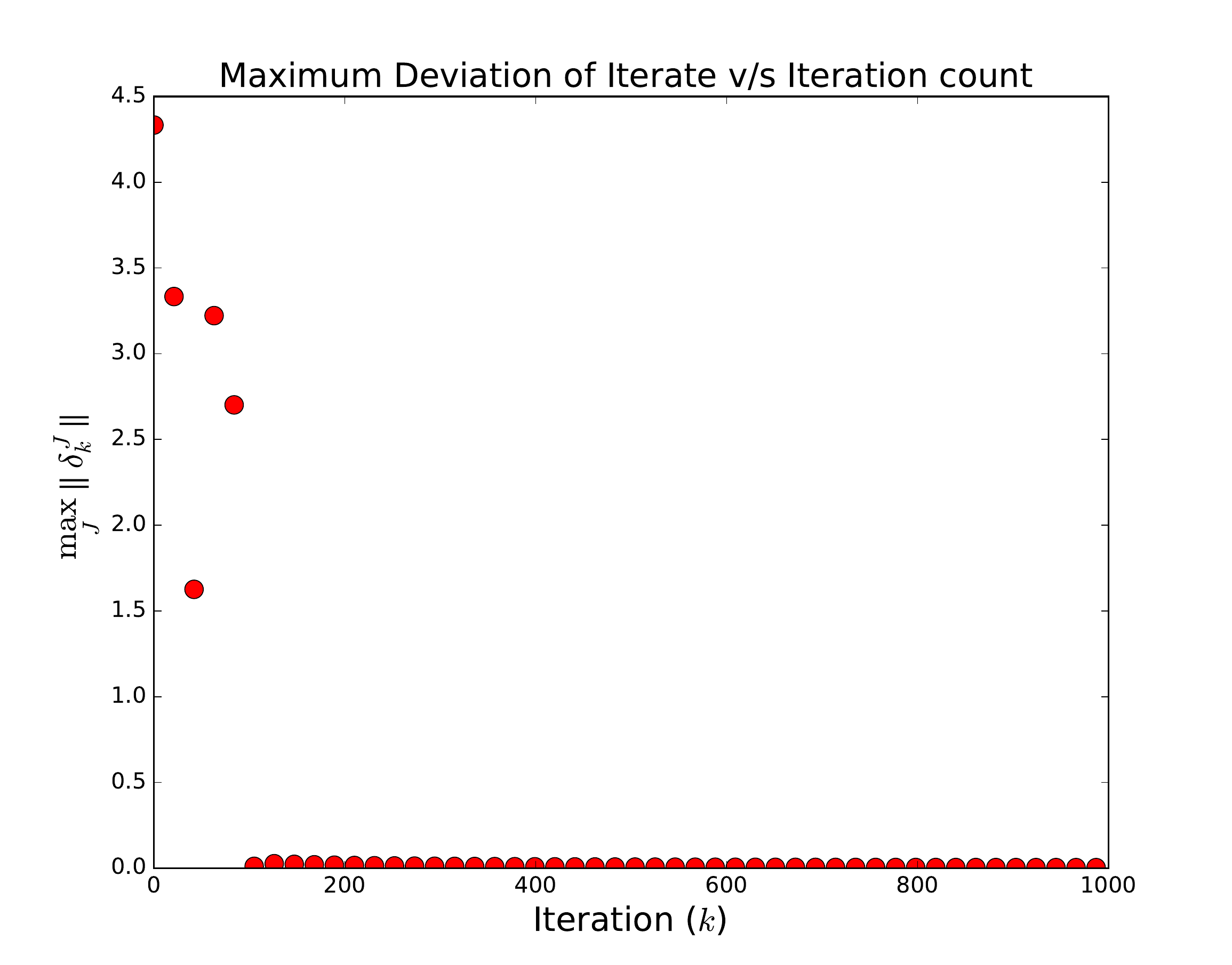}
  \caption{Maximum deviation of iterate from iterate-average.}
  \label{Fig:Sim2}
\end{subfigure} \\
\begin{subfigure}{.48\textwidth}
  \centering
  \includegraphics[width=1.1\linewidth]{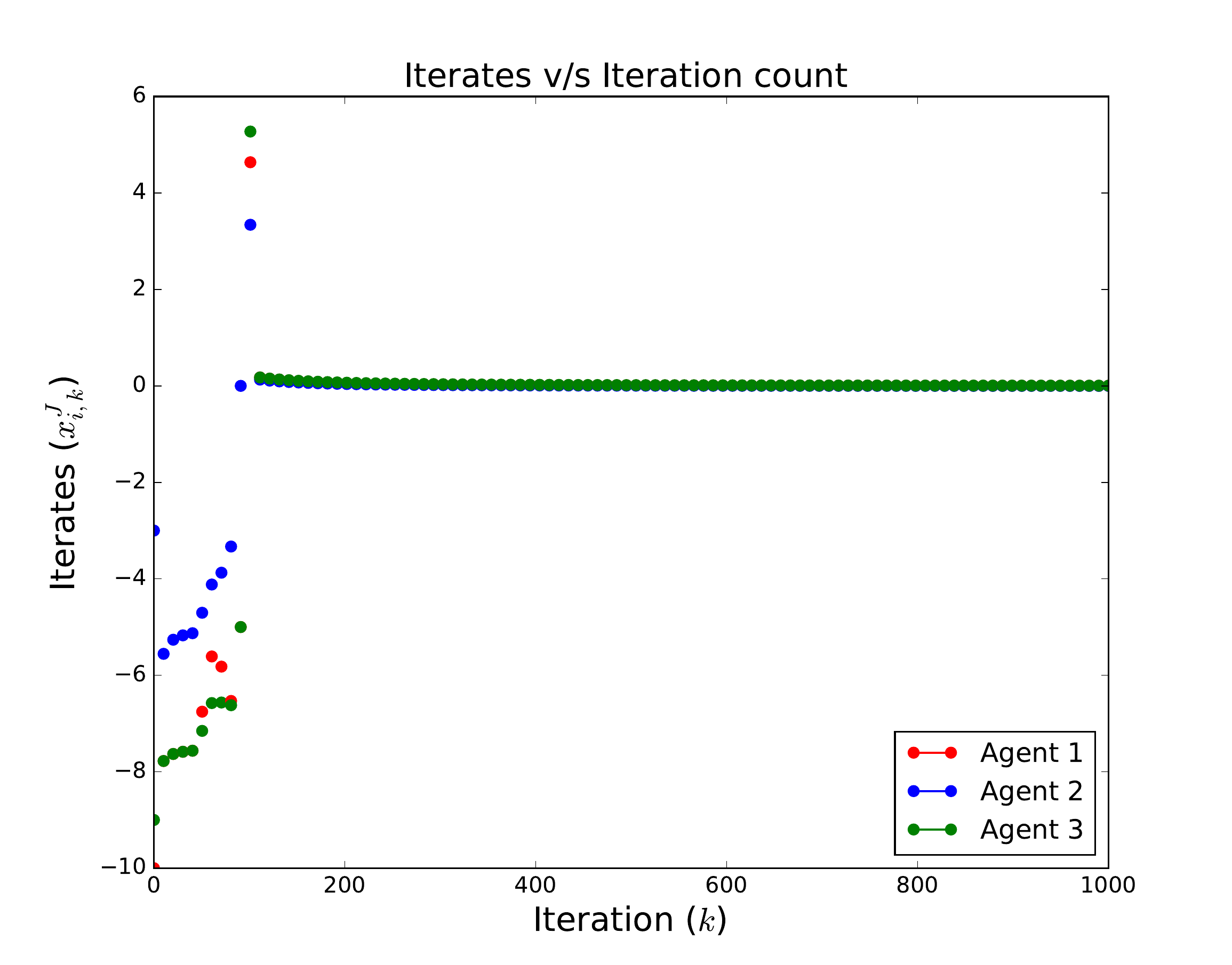}
  \caption{Iterate at Agents 1, 2 and 3.}
  \label{Fig:Sim3}
\end{subfigure} \hfill
\begin{subfigure}{.48\textwidth}
  \centering
  \includegraphics[width=1.1\linewidth]{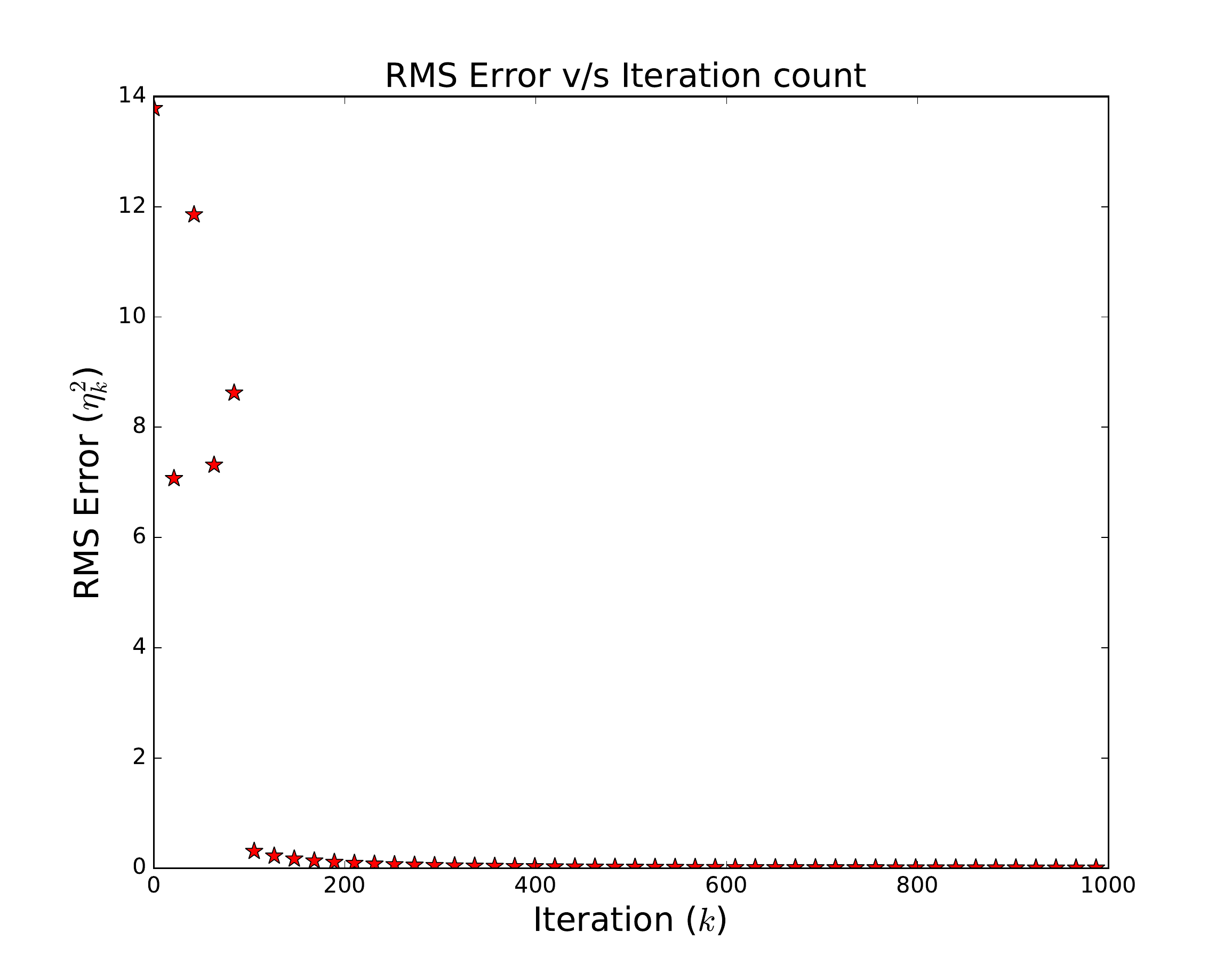}
  \caption{RMS Error $\eta^2_k$}
  \label{Fig:Sim4}
\end{subfigure}
\caption{Correctness and Convergence of Privacy Preserving Algorithm~\ref{Algo:PrivDistOptNCFun}.}
\label{Fig:Simulation}
\end{figure}

The iterate-average ($\bar{x}_k = (x^1_k + x^2_k + x^3_k)/3$) is plotted in Figure~\ref{Fig:Sim1}. The iterate average converges to the optimal point of $f(x)$. Figure~\ref{Fig:Sim2} plots the maximum deviation ($\max \delta_k = \max_J \|x^J_k - \bar{x}_k\|$) of an iterate from the iterate average. The reduction in maximum deviation shows that all the agents converge to the iterate average. Figure~\ref{Fig:Sim3} plots the iterates of each agent with respect to the iteration count. Figure~\ref{Fig:Sim4} decrease in the RMS error ($\eta^2_k = \left[(x^1_k - \bar{x}_k)^2 + (x^2_k - \bar{x}_k)^2 + (x^3_k - \bar{x}_k)^2 \right]$).

\begin{figure}[t]
\begin{subfigure}{.5\textwidth}
  \centering
  \includegraphics[width=1.1\linewidth]{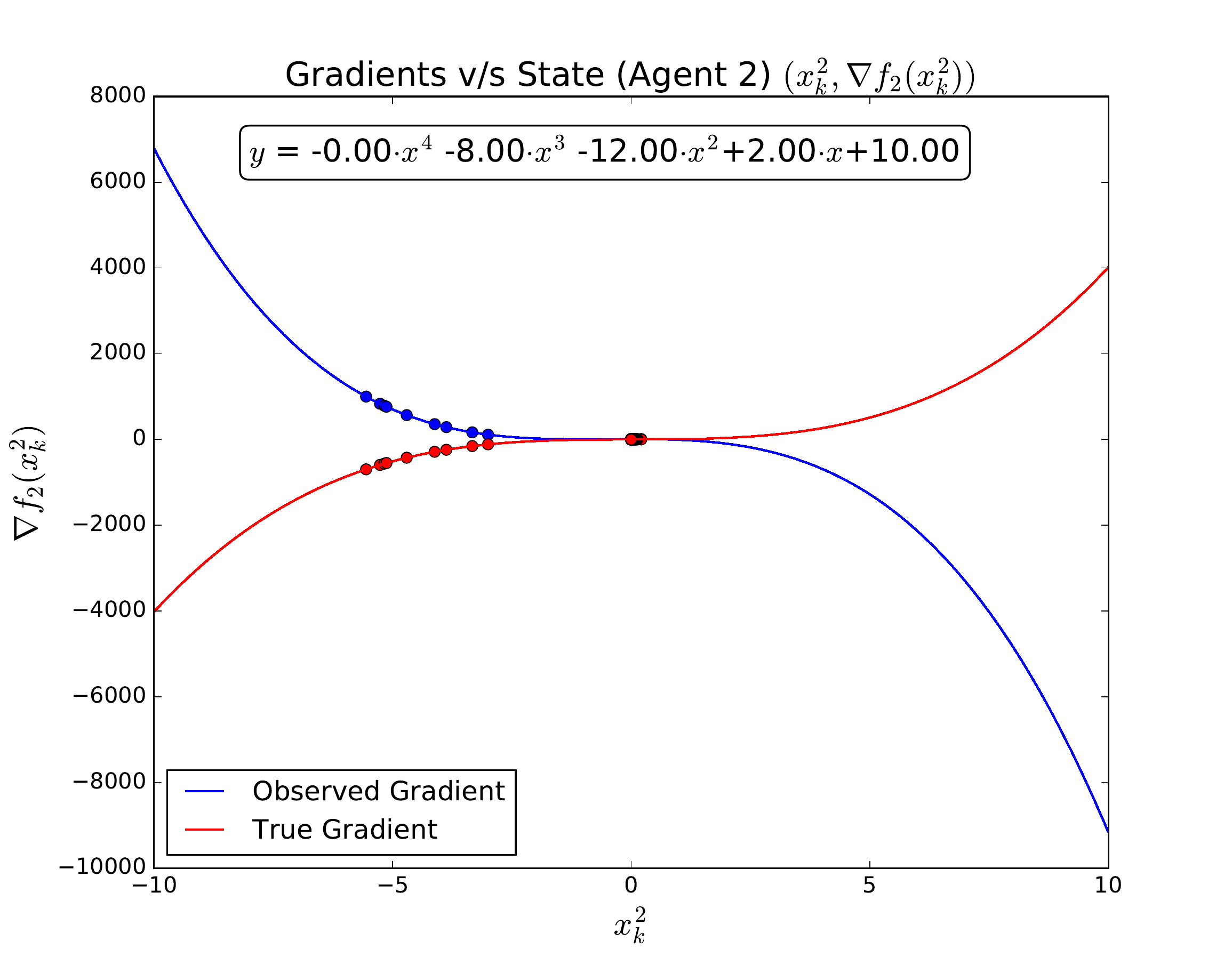}
  \caption{$\nabla f_2(x)$ v/s $\nabla \hat{f}_2(x)$}
  \label{Fig:Sim5}
\end{subfigure} \hfill
\begin{subfigure}{.5\textwidth}
  \centering
  \includegraphics[width=1.1\linewidth]{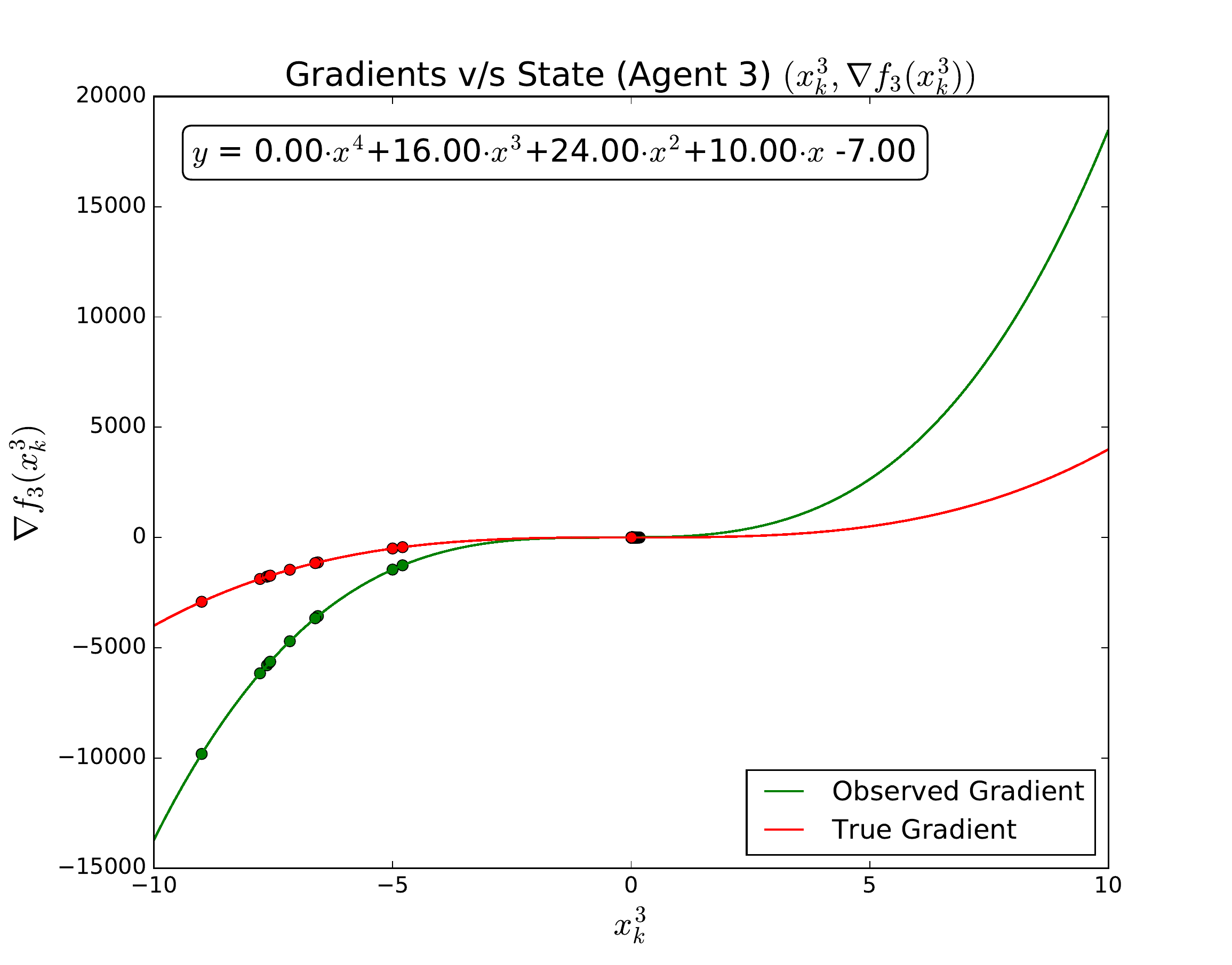}
  \caption{$\nabla f_3(x)$ v/s $\nabla \hat{f}_3(x)$}
  \label{Fig:Sim6}
\end{subfigure}%
\caption{True and Observed Gradients (Least Squares Polynomial Fit $<\texttt{numpy.polyfit}>$)}
\label{Fig:Simulation2}
\end{figure}

Note that the observed gradient is drastically different from the true gradient. This ensures that even a strong PC adversary can estimate only the obfuscated function. The private objective function remains hidden and an adversary cannot uncover any private objective function. 

\section{Conclusion}
In this report, we present a privacy preserving algorithm for distributed optimization over networks. We show that a secure function sharing approach can be used to obfuscate private objective functions and protect them from adversarial estimation. We define privacy of individual objective functions, in the context of distributed optimization problems, as the inability of an adversary to reduce ambiguity between any arbitrarily guessed (feasible) candidate function. We use vertex connectivity ($> f$) to characterize graphs that allow Algorithm~\ref{Algo:PrivDistOptNCFun} to be privacy preserving against a coalition of $f$ adversaries, in the sense that any arbitrary guess of the objective functions made by the coalition is equally consistent. We prove that the function sharing strategy for privacy preserving distributed optimization preserves privacy under an $f$-admissible graph topology and show correctness, convergence and privacy claims through an example and simulations.

{\small
\bibliography{Central.bib}

\begin{thebibliography}{10}

\bibitem{abadi2016deep}
M.~Abadi, A.~Chu, I.~Goodfellow, H.~B. McMahan, I.~Mironov, K.~Talwar, and
  L.~Zhang, ``Deep learning with differential privacy,'' {\em arXiv preprint
  arXiv:1607.00133}, 2016.

\bibitem{boyd2011distributed}
S.~Boyd, N.~Parikh, E.~Chu, B.~Peleato, and J.~Eckstein, ``Distributed
  optimization and statistical learning via the alternating direction method of
  multipliers,'' {\em Foundations and Trends{\textregistered} in Machine
  Learning}, vol.~3, no.~1, pp.~1--122, 2011.

\bibitem{kraska2013mlbase}
T.~Kraska, A.~Talwalkar, J.~C. Duchi, R.~Griffith, M.~J. Franklin, and M.~I.
  Jordan, ``Mlbase: A distributed machine-learning system.,'' in {\em CIDR},
  vol.~1, pp.~2--1, 2013.

\bibitem{NIPS2014_5597}
M.~Li, D.~G. Andersen, A.~J. Smola, and K.~Yu, ``Communication efficient
  distributed machine learning with the parameter server,'' in {\em Advances in
  Neural Information Processing Systems 27}, pp.~19--27, Curran Associates,
  Inc., 2014.

\bibitem{Li2013b}
M.~Li, L.~Zhou, Z.~Yang, A.~Li, and F.~Xia, ``{Parameter Server for Distributed
  Machine Learning},'' {\em Big Learning Workshop}, pp.~1--10, 2013.

\bibitem{cano2016towards}
I.~Cano, M.~Weimer, D.~Mahajan, C.~Curino, and G.~M. Fumarola, ``Towards
  geo-distributed machine learning,'' {\em arXiv preprint arXiv:1603.09035},
  2016.

\bibitem{Nedi2015a}
A.~Nedi{\'{c}} and A.~Olshevsky, ``{Distributed Optimization Over Time-Varying
  Directed Graphs},'' {\em IEEE Transactions on Automatic Control}, vol.~60,
  no.~3, pp.~601--615, 2015.

\bibitem{gade16distoptclientserver}
S.~Gade and N.~Vaidya, ``Distributed optimization for client-server
  architecture with negative gradient weights,'' {\em arXiv preprint
  arXiv:1608.03866}, 2016.

\bibitem{nesterov2012efficiency}
Y.~Nesterov, ``Efficiency of coordinate descent methods on huge-scale
  optimization problems,'' {\em SIAM Journal on Optimization}, vol.~22, no.~2,
  pp.~341--362, 2012.

\bibitem{liu2015asynchronous}
J.~Liu and S.~J. Wright, ``Asynchronous stochastic coordinate descent:
  Parallelism and convergence properties,'' {\em SIAM Journal on Optimization},
  vol.~25, no.~1, pp.~351--376, 2015.

\bibitem{agarwal2011distributed}
A.~Agarwal and J.~C. Duchi, ``Distributed delayed stochastic optimization,'' in
  {\em Advances in Neural Information Processing Systems}, pp.~873--881, 2011.

\bibitem{Singh2014}
C.~Singh, A.~Nedi{\'{c}}, and R.~Srikant, ``{Random Block-Coordinate Gradient
  Projection Algorithms},'' in {\em 53rd IEEE Conference on Decision and
  Control}, pp.~185--190, 2014.

\bibitem{Nedic2007}
A.~Nedi{\'{c}} and A.~Ozdaglar, ``{On the rate of convergence of distributed
  subgradient methods for multi-agent optimization},'' {\em Proceedings of the
  IEEE Conference on Decision and Control}, vol.~54, no.~1, pp.~4711--4716,
  2007.

\bibitem{nedic2011asynchronous}
A.~Nedi{\'{c}}, ``Asynchronous broadcast-based convex optimization over a
  network,'' {\em IEEE Transactions on Automatic Control}, vol.~56, no.~6,
  pp.~1337--1351, 2011.

\bibitem{Nedi2015}
A.~Nedi{\'{c}} and A.~Olshevsky, ``{Distributed Optimization Over Time-Varying
  Directed Graphs},'' {\em IEEE Transactions on Automatic Control}, vol.~60,
  no.~3, pp.~601--615, 2015.

\bibitem{zhang2014asynchronous}
R.~Zhang and J.~T. Kwok, ``Asynchronous distributed admm for consensus
  optimization.,'' in {\em ICML}, pp.~1701--1709, 2014.

\bibitem{huang2015differentially}
Z.~Huang, S.~Mitra, and N.~Vaidya, ``Differentially private distributed
  optimization,'' in {\em Proceedings of the 2015 International Conference on
  Distributed Computing and Networking}, p.~4, ACM, 2015.

\bibitem{rabbat2004distributed}
M.~Rabbat and R.~Nowak, ``Distributed optimization in sensor networks,'' in
  {\em Proceedings of the 3rd international symposium on Information processing
  in sensor networks}, pp.~20--27, ACM, 2004.

\bibitem{zhu2012distributed}
M.~Zhu and S.~Mart{\'\i}nez, ``On distributed convex optimization under
  inequality and equality constraints,'' {\em IEEE Transactions on Automatic
  Control}, vol.~57, no.~1, pp.~151--164, 2012.

\bibitem{sra2012optimization}
S.~Sra, S.~Nowozin, and S.~J. Wright, {\em Optimization for machine learning}.
\newblock MIT Press, 2012.

\bibitem{hadjicostis2016robust}
C.~N. Hadjicostis, N.~H. Vaidya, and A.~D. Dom{\'\i}nguez-Garc{\'\i}a, ``Robust
  distributed average consensus via exchange of running sums,'' {\em IEEE
  Transactions on Automatic Control}, vol.~61, no.~6, pp.~1492--1507, 2016.

\bibitem{wei20131}
E.~Wei and A.~Ozdaglar, ``On the o (1= k) convergence of asynchronous
  distributed alternating direction method of multipliers,'' in {\em Global
  Conference on Signal and Information Processing (GlobalSIP), 2013 IEEE},
  pp.~551--554, IEEE, 2013.

\bibitem{ram2010distributed}
S.~S. Ram, A.~Nedi{\'c}, and V.~V. Veeravalli, ``Distributed stochastic
  subgradient projection algorithms for convex optimization,'' {\em Journal of
  optimization theory and applications}, vol.~147, no.~3, pp.~516--545, 2010.

\bibitem{ram2009incremental}
S.~S. Ram, A.~Nedi{\'{c}}, and V.~V. Veeravalli, ``Incremental stochastic
  subgradient algorithms for convex optimization,'' {\em SIAM Journal on
  Optimization}, vol.~20, no.~2, pp.~691--717, 2009.

\bibitem{su2016fault}
L.~Su and N.~H. Vaidya, ``Fault-tolerant multi-agent optimization: Optimal
  iterative distributed algorithms,'' in {\em Proceedings of the 2016 ACM
  Symposium on Principles of Distributed Computing}, pp.~425--434, ACM, 2016.

\bibitem{shokri2015}
R.~Shokri and V.~Shmatikov, ``{Privacy-Preserving Deep Learning},'' {\em
  Proceedings of the 22nd ACM SIGSAC Conference on Computer and Communications
  Security - CCS '15}, pp.~1310--1321, 2015.

\bibitem{abbe2012privacy}
E.~A. Abbe, A.~E. Khandani, and A.~W. Lo, ``Privacy-preserving methods for
  sharing financial risk exposures,'' {\em The American Economic Review},
  vol.~102, no.~3, pp.~65--70, 2012.

\bibitem{alizadeh2012grid}
M.~Alizadeh, T.-H. Chang, and A.~Scaglione, ``Grid integration of distributed
  renewables through coordinated demand response,'' in {\em 2012 IEEE 51st IEEE
  Conference on Decision and Control (CDC)}, pp.~3666--3671, IEEE, 2012.

\bibitem{pasqualetti2012cyber}
F.~Pasqualetti, F.~D{\"o}rfler, and F.~Bullo, ``Cyber-physical security via
  geometric control: Distributed monitoring and malicious attacks,'' in {\em
  2012 IEEE 51st IEEE Conference on Decision and Control (CDC)},
  pp.~3418--3425, IEEE, 2012.

\bibitem{shokri2015privacy}
R.~Shokri and V.~Shmatikov, ``Privacy-preserving deep learning,'' in {\em
  Proceedings of the 22nd ACM SIGSAC Conference on Computer and Communications
  Security}, pp.~1310--1321, ACM, 2015.

\bibitem{SwiftKey}
``{SwiftKey}.'' \url{https://swiftkey.com/en}.
\newblock Accessed: 2016-10-16.

\bibitem{mcmahan2016federated}
H.~B. McMahan, E.~Moore, D.~Ramage, {\em et~al.}, ``Federated learning of deep
  networks using model averaging,'' {\em arXiv preprint arXiv:1602.05629},
  2016.

\bibitem{hippert2001neural}
H.~S. Hippert, C.~E. Pedreira, and R.~C. Souza, ``Neural networks for
  short-term load forecasting: A review and evaluation,'' {\em IEEE
  Transactions on power systems}, vol.~16, no.~1, pp.~44--55, 2001.

\bibitem{vaidya2009privacy}
J.~Vaidya, ``Privacy-preserving linear programming,'' in {\em Proceedings of
  the 2009 ACM symposium on Applied Computing}, pp.~2002--2007, ACM, 2009.

\bibitem{hong2016privacy}
Y.~Hong, J.~Vaidya, N.~Rizzo, and Q.~Liu, ``Privacy preserving linear
  programming,'' {\em arXiv preprint arXiv:1610.02339}, 2016.

\bibitem{cortes2005spatially}
J.~Cortes, S.~Martinez, and F.~Bullo, ``Spatially-distributed coverage
  optimization and control with limited-range interactions,'' {\em ESAIM:
  Control, Optimisation and Calculus of Variations}, vol.~11, no.~4,
  pp.~691--719, 2005.

\bibitem{ren2005survey}
W.~Ren, R.~W. Beard, and E.~M. Atkins, ``A survey of consensus problems in
  multi-agent coordination,'' in {\em Proceedings of the 2005, American Control
  Conference, 2005.}, pp.~1859--1864, IEEE, 2005.

\bibitem{olfati2007distributed}
R.~Olfati-Saber, ``Distributed tracking for mobile sensor networks with
  information-driven mobility,'' in {\em 2007 American Control Conference},
  pp.~4606--4612, IEEE, 2007.

\bibitem{gade2013heterogeneous}
S.~Gade and A.~Joshi, ``Heterogeneous uav swarm system for target search in
  adversarial environment,'' in {\em 2013 International Conference on Control
  Communication and Computing (ICCC)}, pp.~358--363, Dec 2013.

\bibitem{gade2015herding}
S.~Gade, A.~A. Paranjape, and S.-J. Chung, ``Herding a flock of birds
  approaching an airport using an unmanned aerial vehicle,'' in {\em AIAA
  Guidance, Navigation, and Control Conference}, p.~1540, AIAA, 2015.

\bibitem{weeraddana2013per}
P.~C. Weeraddana, G.~Athanasiou, C.~Fischione, and J.~S. Baras, ``Per-se
  privacy preserving solution methods based on optimization,'' in {\em
  Proceedings of the 52nd IEEE Conference on Decision and Control (CDC)},
  pp.~206--211, 2013.

\bibitem{weeraddana2012privacy}
P.~C. Weeraddana, G.~Athanasiou, M.~Jakobsson, C.~Fischione, and J.~S. Baras,
  ``On the privacy of optimization approaches,'' {\em arXiv preprint
  arXiv:1210.3283}, 2012.

\bibitem{goldreich2009foundations}
O.~Goldreich, {\em Foundations of cryptography: volume 2, basic applications}.
\newblock Cambridge university press, 2009.

\bibitem{goldwasser1997multi}
S.~Goldwasser, ``Multi party computations: past and present,'' in {\em
  Proceedings of the sixteenth annual ACM symposium on Principles of
  distributed computing}, pp.~1--6, ACM, 1997.

\bibitem{pinkas2002cryptographic}
B.~Pinkas, ``Cryptographic techniques for privacy-preserving data mining,''
  {\em ACM Sigkdd Explorations Newsletter}, vol.~4, no.~2, pp.~12--19, 2002.

\bibitem{bednarz2012methods}
A.~Bednarz, {\em Methods for two-party privacy-preserving linear programming.}
\newblock PhD thesis, 2012.

\bibitem{7431982}
S.~Han, U.~Topcu, and G.~J. Pappas, ``Differentially private distributed
  constrained optimization,'' {\em IEEE Transactions on Automatic Control},
  vol.~PP, no.~99, pp.~1--1, 2016.

\bibitem{nozari2015differentially}
E.~Nozari, P.~Tallapragada, and J.~Cort{\'e}s, ``Differentially private
  distributed convex optimization via functional perturbation,'' {\em arXiv
  preprint arXiv:1512.00369}, 2015.

\bibitem{mangasarian2012privacy}
O.~L. Mangasarian, ``Privacy-preserving horizontally partitioned linear
  programs,'' {\em Optimization Letters}, vol.~6, no.~3, pp.~431--436, 2012.

\bibitem{wang2011secure}
C.~Wang, K.~Ren, and J.~Wang, ``Secure and practical outsourcing of linear
  programming in cloud computing,'' in {\em INFOCOM, 2011 Proceedings IEEE},
  pp.~820--828, IEEE, 2011.

\bibitem{dreier2011practical}
J.~Dreier and F.~Kerschbaum, ``Practical privacy-preserving multiparty linear
  programming based on problem transformation,'' in {\em Privacy, Security,
  Risk and Trust (PASSAT) and 2011 IEEE Third Inernational Conference on Social
  Computing (SocialCom), 2011 IEEE Third International Conference on},
  pp.~916--924, IEEE, 2011.

\bibitem{lou2015privacy}
Y.~Lou, L.~Yu, and S.~Wang, ``Privacy preservation in distributed subgradient
  optimization algorithms,'' {\em arXiv preprint arXiv:1512.08822}, 2015.

\bibitem{gade16convsum}
S.~Gade and N.~Vaidya, ``Distributed optimization of convex sum of non-convex
  functions,'' {\em arXiv preprint arXiv:1608.05401}, 2016.

\bibitem{lamport1982byzantine}
L.~Lamport, R.~Shostak, and M.~Pease, ``The byzantine generals problem,'' {\em
  ACM Transactions on Programming Languages and Systems (TOPLAS)}, vol.~4,
  no.~3, pp.~382--401, 1982.

\bibitem{shamir1979share}
A.~Shamir, ``How to share a secret,'' {\em Communications of the ACM}, vol.~22,
  no.~11, pp.~612--613, 1979.

\bibitem{li2015differentially}
C.~Li and P.~Zhou, ``Differentially private distributed online learning,'' {\em
  arXiv preprint arXiv:1505.06556}, 2015.

\bibitem{bertsekas2003convex}
D.~P. Bertsekas, A.~Nedi{\'{c}}, A.~E. Ozdaglar, {\em et~al.}, {\em Convex
  analysis and optimization}.
\newblock Athena Scientific, 2003.

\bibitem{bertsekas1976goldstein}
D.~P. Bertsekas, ``{On the Goldstein-Levitin-Polyak gradient projection
  method},'' {\em Automatic Control, IEEE Transactions on}, vol.~21, no.~2,
  pp.~174--184, 1976.

\bibitem{judson2010abstract}
T.~W. Judson, ``Abstract algebra,'' 2010.

\bibitem{gallian2016contemporary}
J.~Gallian, {\em Contemporary abstract algebra}.
\newblock Cengage Learning, 2016.

\bibitem{godsil2013algebraic}
C.~Godsil and G.~F. Royle, {\em Algebraic graph theory}, vol.~207.
\newblock Springer Science \& Business Media, 2013.

\bibitem{whitney1932congruent}
H.~Whitney, ``Congruent graphs and the connectivity of graphs,'' {\em American
  Journal of Mathematics}, vol.~54, no.~1, pp.~150--168, 1932.

\bibitem{diestel2005graph}
R.~Diestel, ``Graph theory. 2005,'' {\em Grad. Texts in Math}, vol.~101, 2005.

\bibitem{bapat2010graphs}
R.~B. Bapat, {\em {Graphs and Matrices}}.
\newblock Springer, 2010.

\bibitem{zelazo2007agreement}
D.~Zelazo, A.~Rahmani, and M.~Mesbahi, ``Agreement via the edge laplacian,'' in
  {\em Decision and Control, 2007 46th IEEE Conference on}, pp.~2309--2314,
  IEEE, 2007.

\end{thebibliography}
\bibliographystyle{ieeetr}}
\end{document}